\documentclass[final,1p,times]{elsarticle}
\journal{Communications in Mathematical Physics}
\usepackage{tikz-cd}
\usepackage{extpfeil}
\usepackage{lmodern}
\usepackage{stmaryrd}
\usepackage{textcomp}
\DeclareMathOperator{\Span}{span}
\usepackage{amssymb}
\usepackage{amsthm}
\usepackage{latexsym}
\usepackage{amsmath}
\usepackage{color}
\usepackage{graphicx}
\usepackage{indentfirst}
\usepackage{mathrsfs}
\usepackage{lipsum}
\usepackage{comment}
\usepackage{physics}
\usepackage{subcaption} 
\allowdisplaybreaks

\newtheorem{theorem}{\color{black}\indent \textbf{Theorem}}[section]
\newtheorem{lemma}{\color{black}\indent Lemma}[section]
\newtheorem{proposition}{\color{black}\indent Proposition}[section]
\newtheorem{definition}{\color{black}\indent Definition}[section]
\newtheorem{remark}{\color{black}\indent Remark}[section]
\newtheorem{corollary}{\color{black}\indent Corollary}[section]
\newtheorem{example}{\color{black}\indent Example}[section]

\newcommand{\iotaop}{\iota}
\newcommand{\R}{\mathbb{R}}

\begin{document}
	
	\begin{frontmatter}

		\title{q-Cosymplectic  Geometry, Integrability and  Reduction}
\author{
Melvin Leok$^{a}$,
Cristina Sardón$^{b}$,
Xuefeng Zhao$^{c}$ \\[1ex]
$^{a}$Department of Mathematics, University of California, San Diego, 
9500 Gilman Drive, Dept. 0112, La Jolla, CA 92093-0112, USA \\
$^{b}$Department of Applied Mathematics, Universidad Politécnica de Madrid, 
Av. Juan de Herrera 6, 28040, Madrid, Spain \\
$^{c}$College of Mathematics, Jilin University, Changchun, 130012, P. R. China \\[1ex]
\texttt{mleok@ucsd.edu}, \texttt{mariacristina.sardon@upm.es}, \texttt{zhaoxuef@jlu.edu.cn}
}

% Corresponding author footnote
\date{}
\thanks{*Corresponding author at: Department of Applied Mathematics, 
Universidad Politécnica de Madrid, C/ José Gutiérrez Abascal, 2, 
28006, Madrid, Spain}

		\begin{abstract}
			In the present paper, we define the concept of a \( q \)-cosymplectic manifold, on which we study the Hamiltonian, gradient, local gradient, and \( q \)-evolution vector fields. Several Liouville--Arnold-type theorems and a \( q \)-cosymplectic Marsden--Weinstein reduction theorem are established. We also provide physical examples illustrating the application of the structure to multitime dynamics (Fast-slow dynamical system). To make our work more self-contained, we include detailed proofs for some results that may resemble those known for cosymplectic manifolds.

		\end{abstract}
		
		\begin{keyword}
			q-Cosymplectic manifold, Symmetry, Hamiltonian system, First integral.
		\end{keyword}
	\end{frontmatter}
	\section{Introduction}

In this article, we introduce the concept of a 
\( q \)-cosymplectic manifold and present the related properties of its 
\( q \)-cosymplectic structure. We also introduce the notions of gradient vector fields, local gradient vector fields, Hamiltonian vector fields, and 
\( q \)-evolution vector fields on such a manifold, and establish a corresponding Liouville--Arnold-type integrability theorem and a \( q \)-cosymplectic Marsden--Weinstein reduction theorem. Additionally, we explore in detail the relationship between symplectic manifolds and 
\( q \)-cosymplectic manifolds. 
	
Symplectic geometry has been effectively utilized in the study of mechanical systems in physics \cite{Arnold,Bogoyavlenskij,Bogoyavlenskij2,Guillemin,Zhao1,Zhao2}. Early contributions to this field date back to Lagrange and Poisson, who investigated the dynamics of rigid bodies and celestial mechanics \cite{Lagrange,Marle}. As is well known, the classical framework of Hamiltonian dynamics is a symplectic manifold \( (M,\Omega) \), where \( \Omega \) is a non-degenerate, closed 2-form \cite{Libermann,Ortega,Zhao3}. The associated Hamilton’s equations are the equations of the flow of a Hamiltonian vector field \( X_H \), determined 
by a Hamiltonian function \( H \) via 
the correspondence \( i_{X_H}\Omega = dH \). Among the class of Hamiltonian systems, the subclass of completely integrable Hamiltonian systems plays a central role.  
The study of completely integrable Hamiltonian systems on symplectic manifolds, which admit a complete sequence of first integrals, began with the pioneering work of Liouville in 1855 on finding local solutions by quadratures \cite{Liouville}. Today, the geometry of such systems is well understood, thanks to the modern formulation by Arnold \cite{Arnold3}, who gave the Liouville--Arnold theorem.

A Liouville--Arnold integrable Hamiltonian system on a \( 2n \)-dimensional symplectic manifold is defined by \( n - 1 \) additional first integrals \( f_i \) with the property that each integral (including \( H \)) is preserved by the Hamiltonian flow of the other 
integrals. This condition is classically known as the involutivity of the first integrals and can be expressed in terms 
of the Poisson bracket as \( \{f_i, f_j\} = 0 \). Since the twentieth century, the Marsden--Weinstein reduction theorem \cite{Marsden} has played a crucial role in describing Hamiltonian systems on manifolds admitting a Lie group of symmetries of the Hamiltonian and the symplectic manifold.

As is well known, the description of time-dependent mechanical systems cannot be directly approached using symplectic geometry \cite{Abraham,Albert,Cappelletti}. Instead, it is possible to modify symplectic geometry to accommodate them. For instance, one can use cosymplectic geometry \cite{Albert,Cappelletti,deon,Libermann2,Libermann3} to deal with time-dependent Hamiltonians within the symplectic framework \cite{deLucas,Zawora}. In the cosymplectic setting, time-dependent Hamiltonian systems are described via a closed differential two-form \( \Omega \) and a closed, non-vanishing one-form \( \lambda \), both defined on a manifold \( M \), such that \( \ker \Omega \oplus \ker \lambda = TM \). Hence, \( M \) is odd-dimensional, and \( (M, \Omega, \lambda) \) is called a \emph{cosymplectic manifold}.

Beyond the philosophical inquiry into the existence of multiple time dimensions, multi-time phenomena arise naturally in various applied and theoretical contexts. In communication theory, for example, network scheduling often necessitates the consideration of distinct temporal layers. Similarly, traffic flow models may involve multiple time scales, particularly in the presence of localized disruptions such as traffic jams. In such systems, processes initiated from a common configuration may evolve according to different optimization principles, with the coupling mechanism determined solely by the initial data. In this direction, one can refer to the work \cite{Gu}, which examines traffic flow problems on inhomogeneous lattices.

From a mathematical standpoint, traffic flow serves as a classical source of conservation laws, exemplified by the derivation of the Burgers equation. Additional instances of multi-time structures occur in general relativity, electromagnetism, and the Kepler problem, where the temporal framework often extends beyond a single evolution parameter. For further discussion, the reader is referred to \cite{Neagu,Stickforth}. A particularly striking example of theories involving multiple time dimensions is found in string theory, which, in some formulations, accommodates more than two temporal parameters; see \cite{Evans,Zwiebach}. For additional theories involving multi-time dependence, one may consult \cite{Bazan,Cardin,Davini}.

In the present paper, we study multi-time dependent mechanical systems defined on a $q$-cosymplectic manifold. Specifically, we utilize a closed differential two-form $\Omega$ together with $q$ closed, nowhere-vanishing one-forms $\lambda_i$, for $i = 1, \dots, q$, to describe the $q$-cosymplectic structure, which will be precisely defined later. We also investigate certain classes of systems on such manifolds, including Hamiltonian systems and evolution systems, focusing on properties such as integrability and Marsden–Weinstein reduction. Furthermore, a concrete application (fast-slow dynamical system)  of a multi-time system is provided to illustrate the theory.

An outline of the paper is as follows. In Section 2, we introduce the concept of a $q$-cosymplectic manifold, explore several of its intrinsic properties, and review some fundamental theorems which we will use. In Section 3, we define Hamiltonian vector fields, gradient vector fields, and $q$-evolution vector fields on $q$-cosymplectic manifolds. We study their key properties and introduce a Poisson bracket on the space of smooth functions defined on a $q$-cosymplectic manifold. Furthermore, we establish several integrability theorems and investigate the mutual relations and transformations between $q$-cosymplectic manifolds and symplectic manifolds.
In Section 4, we present the Marsden–Weinstein reduction theorem on $q$-cosymplectic manifolds. In Section 5, we discuss the application of the proposed $q$-cosymplectic manifold framework to slow-fast Hamiltonian systems.

    \section{$q$-cosymplectic manifold and Liouville torus }
\begin{definition}[\( q \)-cosymplectic manifolds]
Let \( n, q \) be positive integers, and let \( M \) be a smooth manifold of dimension \( 2n + q \). A \( q \)-cosymplectic structure on \( M \) is a collection consisting of a closed 2-form \( \Omega \) and a tuple of \( q \) (pointwise) linearly independent, non-vanishing, closed 1-forms \( \vec{\lambda} = (\lambda_1, \dots, \lambda_q) \), together with a splitting
\[
TM = \mathcal{R} \oplus \xi
\]
of the tangent bundle, satisfying the following conditions:

\begin{itemize}
    \item[(i)] \( \xi := \bigcap_{i=1}^q \ker \lambda_i; \)
    
    \item[(ii)] There exists a unique collection of linearly independent vector fields \( R_1, \dots, R_q \), tangent to \( \mathcal{R} \), satisfying the relations
    \[
    \lambda_i(R_j) = \delta_{ij}, \quad \text{for all } i,j = 1, \dots, q;
    \]
    
    \item[(iii)] \( \mathcal{R} = \mathrm{span}\{ R_1, \dots, R_q \}; \)
    
    \item[(iv)] \( \ker \Omega = \mathcal{R} \), and the restriction \( \Omega|_\xi \) is non-degenerate.
\end{itemize}

The vector fields \( R_i \), \( i = 1, \dots, q \), are called the \emph{Reeb vector fields}, and the forms \( \lambda_i \), \( i = 1, \dots, q \), are called the \emph{cosymplectic forms}.
\end{definition}

\begin{remark}
The linear independence of the \( \lambda_i \) implies that \( \xi \) has constant rank \( 2n \). Therefore, by condition (iv), \( (\xi, \Omega) \) defines a symplectic vector bundle over \( M \).
\end{remark}

\begin{proposition}
If the vector fields \( R_1, \dots, R_q \) satisfy conditions \emph{(ii)} and \emph{(iii)}, then
\[
[R_i, R_j] = 0, \quad \text{for all } i,j = 1, \dots, q.
\]
\end{proposition}

\begin{proof}
On the one hand, observe that for any \( R_i, R_j \) and \( \lambda_k \), with \( i,j,k = 1, \dots, q \), we have
\begin{align*}
0 = \mathcal{L}_{R_i}(i_{R_j} \lambda_k) = i_{[R_i, R_j]} \lambda_k + i_{R_j} \mathcal{L}_{R_i} \lambda_k = i_{[R_i, R_j]} \lambda_k,
\end{align*}
which implies that \( [R_i, R_j] \in \xi \).

On the other hand, since
\[
0 = \mathcal{L}_{R_i}(i_{R_j} \Omega) = i_{[R_i, R_j]} \Omega + i_{R_j} \mathcal{L}_{R_i} \Omega = i_{[R_i, R_j]} \Omega,
\]
and \( \Omega \) is non-degenerate on \( \xi \), it follows that \( [R_i, R_j] = 0 \).
\end{proof}

\begin{example}
Let \( T_i \), \( i = 1, \dots, q \), be one-dimensional manifolds, and let \( (P, \Omega) \) be a symplectic manifold. Then, we can define a \( q \)-cosymplectic structure on the product manifold
\[
M = T_1 \times \cdots \times T_q \times P.
\]
Let \( \pi_{T_i} : M \rightarrow T_i \) and \( \pi_P : M \rightarrow P \) be the projections onto the respective factors. The symplectic form \( \Omega \) on \( P \) induces a closed 2-form on \( M \) given by
\[
\Omega_P := \pi_P^* \Omega.
\]
Similarly, a non-vanishing 1-form \( \eta_i \) on \( T_i \) induces a closed 1-form \( \eta_{iT} := \pi_{T_i}^* \eta_i \) on \( M \). Then,
\[
\left( M = T_1 \times \cdots \times T_q \times P, \; \Omega_P, \; \vec{T} = (\eta_{1T}, \dots, \eta_{qT}) \right)
\]
defines a \( q \)-cosymplectic manifold.
\end{example}

A manifold endowed with such a structure is called a \emph{\( q \)-cosymplectic manifold} and is denoted by \( (M, \Omega, \vec{\lambda}) \), or simply by \( M \) when the context is clear. We refer to the collection \( \{ \lambda_i \} \) as an \emph{adapted coframe} for the \( q \)-cosymplectic structure, and the \( q \)-form
\[
\lambda := \lambda_1 \wedge \cdots \wedge \lambda_q \neq 0
\]
is called the \emph{characteristic form}. The bundles \( \mathcal{R} \) and \( \xi \) are called the \emph{Reeb distribution} and the \emph{\( q \)-cosymplectic distribution}, respectively. The elements of \( \xi \) are referred to as \emph{horizontal vector fields}.

\begin{definition}
Let \( (M, \Omega, \vec{\lambda}) \) be a \( q \)-cosymplectic manifold, and let \( X \) be a vector field on \( M \) with flow \( \psi_t \) (for all \( t \in \mathbb{R} \) such that the flow on \( M \) is well-defined).

\begin{itemize}
    \item We call \( X \) a \emph{\( q \)-cosymplectic vector field} (or \emph{infinitesimal automorphism of \( \xi \)}) if \( T\psi_t(\xi) = \xi \) for all \( t \) for which the flow is defined.

    \item We call \( X \) a \emph{strict \( q \)-cosymplectic vector field} (or \emph{infinitesimal automorphism of \( \lambda_1 \wedge \cdots \wedge \lambda_q \)}) if
    \[
    \psi_t^*(\lambda_1 \wedge \cdots \wedge \lambda_q) = \lambda_1 \wedge \cdots \wedge \lambda_q
    \]
    for all \( t \) such that the flow is well-defined.
\end{itemize}
\end{definition}

The following theorem provides a practical criterion for determining whether a given vector field on a \( q \)-cosymplectic manifold is (strictly) \( q \)-cosymplectic.

\begin{theorem} \label{T1}
Let \( (M, \Omega, \vec{\lambda}) \) be a \( q \)-cosymplectic manifold, and let \( X \) be a vector field on \( M \).

\begin{itemize}
    \item[i)] \( X \) is a \( q \)-cosymplectic vector field if and only if there exists a smooth function \( \mu : M \rightarrow \mathbb{R} \) such that
    \[
    \mathcal{L}_X(\lambda_1 \wedge \cdots \wedge \lambda_q) = \mu \, \lambda_1 \wedge \cdots \wedge \lambda_q,
    \]
    where \( \mathcal{L}_X \) denotes the Lie derivative with respect to \( X \).

    \item[ii)] \( X \) is a strict \( q \)-cosymplectic vector field if and only if
    \[
    \mathcal{L}_X(\lambda_1 \wedge \cdots \wedge \lambda_q) = 0.
    \]
\end{itemize}
\end{theorem}

\begin{proof}
\textit{(i)} Assume that \( T\psi_t(\xi) = \xi \), where \( \psi_t \) is the flow of \( X \). This means that for every \( t \), there exists a nowhere vanishing function \( \eta_t \) such that
\[
\psi_t^*(\lambda_1 \wedge \cdots \wedge \lambda_q) = \eta_t \, (\lambda_1 \wedge \cdots \wedge \lambda_q).
\]
Consequently,
\[
\mathcal{L}_X(\lambda_1 \wedge \cdots \wedge \lambda_q) 
= \left.\frac{d}{dt}\right|_{t=0} \left( \psi_t^*(\lambda_1 \wedge \cdots \wedge \lambda_q) \right) 
= \left.\frac{d}{dt}\right|_{t=0} \left( \eta_t \lambda_1 \wedge \cdots \wedge \lambda_q \right) 
= \mu \, \lambda_1 \wedge \cdots \wedge \lambda_q,
\]
where \( \mu = \left. \frac{d}{dt} \right|_{t=0} \eta_t \).

Conversely, suppose there exists a smooth function \( \mu \) such that
\[
\mathcal{L}_X(\lambda_1 \wedge \cdots \wedge \lambda_q) = \mu \, (\lambda_1 \wedge \cdots \wedge \lambda_q).
\]
This implies (differentiating at an arbitrary \( t \), not necessarily at \( t = 0 \)) that
\begin{align} \label{E1}
\frac{d}{dt} \left( \psi_t^*(\lambda_1 \wedge \cdots \wedge \lambda_q) \right) 
= \psi_t^*(\mu \, \lambda_1 \wedge \cdots \wedge \lambda_q) 
= (\mu \circ \psi_t) \, \psi_t^*(\lambda_1 \wedge \cdots \wedge \lambda_q).
\end{align}

Here we have used the group property of the flow: \( \psi_{t + t_0} = \psi_{t_0} \circ \psi_t \), which implies the identity
\[
\left. \frac{d}{dt} \right|_{t = t_0} \left( \psi_t^*(\lambda_1 \wedge \cdots \wedge \lambda_q) \right)
= \psi_{t_0}^* \left( \left. \frac{d}{dt} \right|_{t = 0} \left( \psi_t^*(\lambda_1 \wedge \cdots \wedge \lambda_q) \right) \right)
= \psi_{t_0}^* \left( \mathcal{L}_X(\lambda_1 \wedge \cdots \wedge \lambda_q) \right).
\]

Viewing equation \eqref{E1} as a differential equation, we can solve it to find
\[
\psi_t^*(\lambda_1 \wedge \cdots \wedge \lambda_q) = \eta_t \, \lambda_1 \wedge \cdots \wedge \lambda_q,
\]
where \( \eta_t \) is given by
\[
\eta_t = \exp\left( \int_0^t (\mu \circ \psi_s) \, ds \right).
\]
Hence, \( X \) is a \( q \)-cosymplectic vector field.

\medskip

\noindent \textit{(ii)} The vector field \( X \) is a strict \( q \)-cosymplectic vector field if and only if its flow preserves the characteristic form \( \lambda_1 \wedge \cdots \wedge \lambda_q \). By the defining property of the Lie derivative, this is equivalent to
\[
\mathcal{L}_X(\lambda_1 \wedge \cdots \wedge \lambda_q) = 0.
\]
\end{proof}

\begin{example}
Let \( (M, \Omega, \vec{\lambda}) \) be a \( q \)-cosymplectic manifold. By Cartan’s magic formula, for every Reeb vector field \( R_i \) (as defined earlier), we have
\[
\mathcal{L}_{R_i} \lambda_j = d(i_{R_i} \lambda_j) + i_{R_i} d\lambda_j = d(\delta_{ij}) + 0 = 0,
\]
and thus
\[
\mathcal{L}_{R_i}(\lambda_1 \wedge \cdots \wedge \lambda_q) 
= (\mathcal{L}_{R_i} \lambda_1) \wedge \lambda_2 \wedge \cdots \wedge \lambda_q + \cdots 
+ \lambda_1 \wedge \cdots \wedge \lambda_{q-1} \wedge (\mathcal{L}_{R_i} \lambda_q) = 0.
\]
Hence, the Reeb vector fields associated with \( \lambda_i \), \( i = 1, \dots, q \), are strict \( q \)-cosymplectic vector fields.
\end{example}

\begin{example}
There exists a simple \( q \)-cosymplectic structure on \( \mathbb{R}^{2n+q} \) with coordinates 
\[
(x_1, y_1, \dots, x_n, y_n, z_1, \dots, z_q),
\]
defined by
\[
\lambda_i := dz_i, \quad \Omega = \sum_{j=1}^n x_j \, dy_j,
\]
and
\[
\mathcal{R} = \mathrm{span} \left\{ \frac{\partial}{\partial z_1}, \dots, \frac{\partial}{\partial z_q} \right\}.
\]
The \( q \)-cosymplectic distribution is given by
\[
\xi = \mathrm{span} \left\{ \frac{\partial}{\partial x_1}, \frac{\partial}{\partial y_1}, \dots, \frac{\partial}{\partial x_n}, \frac{\partial}{\partial y_n} \right\}.
\]
We observe that the vector fields \( \frac{\partial}{\partial z_i} \), \( i = 1, \dots, q \), are Reeb vector fields. Moreover, they are strict \( q \)-cosymplectic vector fields, since
\[
\mathcal{L}_{\frac{\partial}{\partial z_i}} \lambda_j = 0, \quad \forall i,j = 1, \dots, q.
\]
\end{example}

We now state the following theorem:

\begin{theorem}[\cite{Bogoyavlenskij}] \label{T3}
Assume that on a manifold \( M \) there exist:
\begin{enumerate}
    \item A submersion \( F = (F_1, \dots, F_k) : M \rightarrow \mathbb{R}^k \) with compact and connected fibers, for some \( 1 \leq k < \dim M \).
    
    \item A collection of \( n = \dim M - k \) vector fields \( Y_1, \dots, Y_n \) that are everywhere linearly independent, pairwise commuting, and tangent to the fibers of \( F \), i.e.,
    \[
    [Y_i, Y_j] = 0, \quad \mathcal{L}_{Y_i} F_r = 0, \quad \forall i,j = 1, \dots, n, \; r = 1, \dots, k.
    \]
\end{enumerate}
Then:
\begin{itemize}
    \item[i.] The map \( F : M \rightarrow F(M) \subset \mathbb{R}^k \) defines a \( \mathbb{T}^n \)-bundle.
    
    \item[ii.] Any vector field \( X \) on \( M \) satisfying
    \[
    \mathcal{L}_X F_r = 0 \quad \text{and} \quad [X, Y_i] = 0, \quad \forall i = 1, \dots, n, \; r = 1, \dots, k,
    \]
    is conjugate to a constant vector field on \( \mathbb{T}^n \) via each bundle chart of \( F : M \rightarrow F(M) \).
\end{itemize}
\end{theorem}

In the above theorem, the tuple \( (Y_1, \dots, Y_n, F_1, \dots, F_k) \) is called an \emph{integrable system of type \( (n, k) \)} on \( M \). This notation simply emphasizes the commuting flows and first integrals. A system on a manifold \( M \) is called \emph{integrable} if there exists an integrable system \( (Y_1, \dots, Y_n, F_1, \dots, F_k) \) of some type \( (n, k) \) on \( M \) with \( Y_1 = Y \).

Note that the vector fields \( Y_1, \dots, Y_n \) are tangent to the fibers of \( F \). We say that the system \( (Y_1, \dots, Y_n, F_1, \dots, F_k) \) is \emph{regular} at a fiber \( L \) of \( F \) if
\[
Y_1 \wedge \cdots \wedge Y_n \neq 0 \quad \text{and} \quad dF_1 \wedge \cdots \wedge dF_k \neq 0 \quad \text{everywhere on } L.
\]
We say that the system is \emph{proper} if the map \( (F_1, \dots, F_k): M \rightarrow \mathbb{R}^k \) is a proper topological map (i.e., each level set is compact), and the system is regular on almost every fiber.

Given a manifold \( M \), the vector bundles on \( M \) which can be obtained from the tangent and cotangent bundles \( TM, T^*M \), and the trivial bundle \( \mathbb{R} \times M \), by operations such as direct sums and tensor products, will be called \emph{natural vector bundles} over \( M \).

\begin{theorem}[Liouville's theorem \cite{Liouville}] \label{L7}
Assume that the integrable system \( (Y_1, \dots, Y_n, F_1, \dots, F_k) \) is as described in Theorem~\ref{T3} and is regular at a compact level set \( L \) of \( F \). Then, in a tubular neighborhood \( \mathcal{U}(L) \), there exists, up to automorphisms of \( \mathbb{T}^n \), a unique free torus action
\[
\rho : \mathbb{T}^n \times \mathcal{U}(L) \rightarrow \mathcal{U}(L),
\]
which preserves the system (that is, the action preserves each \( Y_i \) and each \( F_j \)), and whose orbits are regular level sets of the system. In particular, \( L \) is diffeomorphic to \( \mathbb{T}^n \), and
\[
\mathcal{U}(L) \cong \mathbb{T}^n \times B^k,
\]
with periodic coordinates \( \theta_1 \, (\mathrm{mod}\, 1), \dots, \theta_n \, (\mathrm{mod}\, 1) \) on the torus \( \mathbb{T}^n \), and coordinates \( (z_1, \dots, z_k) \) on a \( k \)-dimensional ball \( B^k \), such that \( F_1, \dots, F_k \) depend only on the variables \( z_1, \dots, z_k \), and the vector fields \( Y_i \) are of the form
\[
Y_i = \sum_{j=1}^n a_{ij}(z_1, \dots, z_k) \, \frac{\partial}{\partial \theta_j}.
\]
\end{theorem}

A system of coordinates
\[
(\theta_1 \; (\mathrm{mod}\, 1), \dots, \theta_n \; (\mathrm{mod}\, 1), z_1, \dots, z_k)
\]
on \( \mathcal{U}(L) \cong \mathbb{T}^n \times B^k \), as given by the above theorem, is called a \emph{Liouville system of coordinates}. 

Due to the previous theorem, each \( n \)-dimensional compact level set \( L \) of an integrable system of type \( (n, k) \), on which the system is regular, is called a \emph{Liouville torus}, and the torus \( \mathbb{T}^n \)-action in a tubular neighborhood \( \mathcal{U}(L) \) of \( L \) that preserves the system is called the \emph{Liouville torus action}. Notice that this action is uniquely determined by the system, up to an automorphism of \( \mathbb{T}^n \).

\begin{theorem}[Fundamental conservation property {\cite{Zung}}] \label{L8}
Let \( L \) be a Liouville torus of an integrable system \( (Y_1, \dots, Y_n, F_1, \dots, F_k) \) on a manifold \( M \), and let \( \mathcal{G} \in \Gamma(\otimes^h TM \otimes^k T^*M) \) be a tensor field on \( M \) that is preserved by all the vector fields of the system:
\[
\mathcal{L}_{Y_i} \mathcal{G} = 0, \quad \forall i = 1, \dots, n.
\]
Then the Liouville torus \( \mathbb{T}^n \)-action on a tubular neighborhood \( \mathcal{U}(L) \subset M \) also preserves \( \mathcal{G} \).
\end{theorem}

\begin{definition}
An integrable system \( (Y_1, \dots, Y_n, F_1, \dots, F_k) \) on a \( 2m + q \)-dimensional \( q \)-cosymplectic manifold \( (M, \vec{\lambda}, \mathcal{R} \oplus \xi) \) is called a \emph{\( q \)-cosymplectic integrable system} if the vector fields \( Y_1, \dots, Y_n \) are strict \( q \)-cosymplectic vector fields.
\end{definition}

\begin{theorem}
If \( (Y_1, \dots, Y_n, F_1, \dots, F_k) \) is a \( q \)-cosymplectic integrable system on a \( 2m + q \)-dimensional \( q \)-cosymplectic manifold \( (M, \Omega, \vec{\lambda}, \mathcal{R} \oplus \xi) \), then in a neighborhood \( \mathcal{U}(N) \cong \mathbb{T}^n \times B^k \) of any Liouville torus \( N \subset M \), the forms \( \lambda_i \) are \( \mathbb{T}^n \)-invariant under the Liouville coordinate system \( (\theta_i \; (\mathrm{mod}\; 1), z_j) \), and are expressed as
\[
\lambda_i = \sum_{j=1}^n a_{ij}(z) \, d\theta_j + \sum_{j=1}^k b_{ij}(z) \, dz_j, \quad i = 1, \dots, q.
\]
% and \( n \leq m + q = (\mathrm{dim} M + q)/2. \)

% Specially, if \( (M, \vec{\lambda}, \mathcal{R} \oplus \xi) \) is a uniform \( q \)-contact manifold (see Definition \ref{D3}), there is a Liouville coordinate system
% \[
% (\theta_1 \; (\mathrm{mod}\; 1), \dots, \theta_n \; (\mathrm{mod}\; 1), z_1, \dots, z_n, x_1)
% \]
\end{theorem}

\begin{proof}
By Theorem~\ref{L8}, the characteristic form \( \lambda_1 \wedge \cdots \wedge \lambda_q \) is \( \mathbb{T}^n \)-invariant under the Liouville coordinate system \( (\theta_i \; (\mathrm{mod}\; 1), z_j) \). Therefore, the characteristic form must be of the form
\[
\lambda_1 \wedge \cdots \wedge \lambda_q = \sum_{\substack{h + l = q \\ 0 \leq h \leq n, \; 0 \leq l \leq k}} f_{h,l}(z) \, d\theta_{i_1} \wedge \cdots \wedge d\theta_{i_h} \wedge dz_{j_1} \wedge \cdots \wedge dz_{j_l},
\]
for some smooth functions \( f_{h,l}(z) \) depending only on \( z \)-variables.

This implies that each \( \lambda_i \) must be of the form
\[
\lambda_i = \sum_{j=1}^n a_{ij}(z) \, d\theta_j + \sum_{j=1}^k b_{ij}(z) \, dz_j,
\]
for some functions \( a_{ij}(z), b_{ij}(z) \), which confirms that each \( \lambda_i \) is \( \mathbb{T}^n \)-invariant.
%
% We must have \( n \leq m + q \), otherwise \( \lambda_1 \wedge \cdots \wedge \lambda_q \wedge (d\lambda_i)^m = 0 \) for \( i = 1, \dots, q \), because it would not contain the component \( (d\theta_1 \wedge \cdots \wedge d\theta_n) \wedge \cdots \).
\end{proof}

\section{Hamiltonian Systems, Completely Integrability, q-Cosymplectization and Symplectization}
Given a \( q \)-cosymplectic manifold \( (M, \Omega, \vec{\lambda}) \), we consider the vector bundle morphism
\[
b: TM \rightarrow T^*M, \quad v_x \in T_x M \mapsto b(v_x) := \left( i_v \Omega + \sum_{i=1}^q \lambda_i(v) \lambda_i \right)\big|_x, \quad x \in M,
\]
which is a vector bundle isomorphism.

One can observe that \( b \) is induced by the non-degenerate covariant 2-tensor field
\[
\Omega + \sum_{i=1}^q \lambda_i \otimes \lambda_i
\]
on \( M \).

For every \( f \in C^\infty(M) \), we can define four vector fields. The first is:

\medskip

\noindent\textbullet\; A \emph{gradient vector field}, defined by
\begin{align}
\nabla f := b^{-1}(df),
\end{align}
which equivalently satisfies
\[
i_{\nabla f} \Omega = df - \sum_{i=1}^q (i_{\nabla f} \lambda_i)\, \lambda_i.
\]

We can verify that
\[
i_{\nabla f} \lambda_i = R_i(f), \quad i = 1, \dots, q.
\]
Indeed, for any Reeb vector field \( R_j \), we have
\[
0 = i_{R_j} i_{\nabla f} \Omega = i_{R_j} df - i_{R_j} \left( \sum_{i=1}^q (i_{\nabla f} \lambda_i) \lambda_i \right) = R_j(f) - i_{\nabla f} \lambda_j,
\]
which proves the claim.

Moreover, we observe the following properties:
\begin{align*}
&0 = i_{\nabla f} i_{\nabla f} \Omega = \nabla f(f) - \sum_{i=1}^q (i_{\nabla f} \lambda_i)^2 
\quad \Longleftrightarrow \quad \nabla f(f) = \sum_{i=1}^q (i_{\nabla f} \lambda_i)^2 \geq 0, \\
&0 = i_{\nabla f} i_{\nabla f} \lambda_i = \nabla f(R_i f), \quad i = 1, \dots, q.
\end{align*}

Hence, \( \nabla f(f) \geq 0 \), and the functions \( R_i f \), for \( i = 1, \dots, q \), are first integrals of \( \nabla f \).

	In particular, 
\[
\nabla f(f) = 0 \; \Longleftrightarrow \; R_i f = 0, \; i = 1, \dots, q \; \Longleftrightarrow \; i_{\nabla f} \lambda_i = 0, \; i = 1, \dots, q.
\]
In this case, \( f \) is a first integral of \( \nabla f \), and locally there exist functions \( g_1, \dots, g_q \) such that \( dg_i = \lambda_i \) and
\[
\nabla f(g_i) = i_{\nabla f} \lambda_i = 0,
\]
which implies that, locally, \( g_1, \dots, g_q \) are also first integrals of \( \nabla f \).

\medskip

\noindent\textbullet\; A \emph{local gradient vector field}, which satisfies
\begin{align} \label{LG}
d\left( i_X \Omega + \sum_{i=1}^q \lambda_i(X) \lambda_i \right) = 0.
\end{align}

\medskip

\noindent\textbullet\; A \emph{Hamiltonian vector field} \( X_f \), defined by
\begin{align}
X_f := b^{-1} \left( df - \sum_{i=1}^q (R_i f) \lambda_i \right),
\end{align}
which is equivalent to
\[
i_{X_f} \Omega = df - \sum_{i=1}^q (R_i f) \lambda_i.
\]
We can verify that
\[
i_{X_f} \lambda_i = 0, \quad i = 1, \dots, q.
\]
Indeed, by the definition of the map \( b \), we have
\[
b(X_f) = i_{X_f} \Omega + \sum_{i=1}^q \lambda_i(X_f) \lambda_i = df - \sum_{i=1}^q (R_i f) \lambda_i.
\]
Therefore, for any \( R_j \), it follows that
\begin{align*}
\lambda_j(X_f) 
&= i_{R_j} i_{X_f} \Omega + i_{R_j} \left( \sum_{i=1}^q \lambda_i(X_f) \lambda_i \right) \\
&= i_{R_j} df - i_{R_j} \left( \sum_{i=1}^q (R_i f) \lambda_i \right) 
= R_j(f) - R_j(f) = 0.
\end{align*}

Moreover, we observe that \( X_f(f) = 0 \), since
\[
0 = i_{X_f} i_{X_f} \Omega = X_f(f) - \sum_{i=1}^q (R_i f) \lambda_i(X_f) = X_f(f).
\]

Additionally, since the \( \lambda_i \) are closed 1-forms and \( i_{X_f} \lambda_i = 0 \), \( i = 1, \dots, q \), it follows that locally there exist functions \( f_1, \dots, f_q \) such that \( df_i = \lambda_i \) and
\[
i_{X_f} df_i = i_{X_f} \lambda_i = 0.
\]
Thus, locally, \( f_1, \dots, f_q \) are also first integrals of \( X_f \).

\medskip

\noindent\textbullet\; A \emph{\( q \)-evolution vector field} associated to the function \( f \), defined by
\begin{equation}\label{evolutionvf}
E_f := \sum_{i=1}^q R_i + X_f.
\end{equation}

 \begin{remark}
Condition \eqref{LG} is equivalent to
\begin{align} \label{LX}
\mathcal{L}_X \Omega = \sum_{i=1}^q \lambda_i \wedge \mathcal{L}_X \lambda_i.
\end{align}
\end{remark}

\begin{proof}
Indeed, we compute:
\begin{align*}
\mathcal{L}_X \Omega &= i_X d\Omega + d i_X \Omega \\
&= \sum_{i=1}^q \lambda_i \wedge d(i_X \lambda_i) \\
&= \sum_{i=1}^q \lambda_i \wedge \left( d(i_X \lambda_i) + i_X d\lambda_i \right) \\
&= \sum_{i=1}^q \lambda_i \wedge \mathcal{L}_X \lambda_i,
\end{align*}
which proves the claim.
\end{proof}

\begin{definition}
An \emph{infinitesimal automorphism} of the \( q \)-cosymplectic manifold \( (M, \Omega, \vec{\lambda}) \) is a vector field \( X \) on \( M \) such that
\begin{align} \label{Ome}
\mathcal{L}_X \Omega = 0, \quad \mathcal{L}_X \lambda_i = 0, \quad i = 1, \dots, q.
\end{align}
\end{definition}

\begin{remark}
We observe that the Reeb vector fields \( R_i \), \( i = 1, \dots, q \), and Hamiltonian vector fields are all infinitesimal automorphisms. Moreover, every infinitesimal automorphism preserves the volume form
\[
\Omega^n \wedge \lambda_1 \wedge \cdots \wedge \lambda_q,
\]
but the converse is not necessarily true.
\end{remark}

\begin{theorem} \label{LGG}
\begin{itemize}
    \item[(1)] The Lie bracket of an infinitesimal automorphism \( X \) and a local gradient vector field \( Y \) is again a local gradient vector field.
    
    \item[(2)] The Lie bracket of an infinitesimal automorphism \( X \) and a Hamiltonian vector field \( X_f \) is again a Hamiltonian vector field.
\end{itemize}
\end{theorem}

\begin{proof}
Using \eqref{LX} and \eqref{Ome}, we compute:
\begin{align*}
\mathcal{L}_{[X, Y]} \Omega &= \mathcal{L}_X \mathcal{L}_Y \Omega \\
&= \mathcal{L}_X \left( \sum_{i=1}^q \lambda_i \wedge \mathcal{L}_Y \lambda_i \right) \\
&= \sum_{i=1}^q \lambda_i \wedge \mathcal{L}_{[X, Y]} \lambda_i,
\end{align*}
which shows that \( [X, Y] \) satisfies the condition for being a local gradient vector field.

Next, for any Reeb vector field \( R_i \), \( i = 1, \dots, q \), we calculate:
\begin{align*}
0 &= \mathcal{L}_X \left( (-1)^{i-1} \Omega^n \wedge \lambda_1 \wedge \cdots \wedge \widehat{\lambda_i} \wedge \cdots \wedge \lambda_q \right) \\
&= \mathcal{L}_X \left( i_{R_i} (\Omega^n \wedge \lambda_1 \wedge \cdots \wedge \lambda_q) \right) \\
&= i_{[X, R_i]} (\Omega^n \wedge \lambda_1 \wedge \cdots \wedge \lambda_q) + i_{R_i} \mathcal{L}_X (\Omega^n \wedge \lambda_1 \wedge \cdots \wedge \lambda_q) \\
&= i_{[X, R_i]} (\Omega^n \wedge \lambda_1 \wedge \cdots \wedge \lambda_q),
\end{align*}
which implies that \( [X, R_i] = 0 \) for all \( i = 1, \dots, q \).

Now we compute:
\begin{align*}
i_{[X, X_f]} \Omega &= \mathcal{L}_X (i_{X_f} \Omega) \\
&= \mathcal{L}_X \left( df - \sum_{i=1}^q (R_i f) \lambda_i \right) \\
&= d(Xf) - \sum_{i=1}^q (i_{[X, R_i]} df) \lambda_i - \sum_{i=1}^q (R_i (Xf)) \lambda_i - \sum_{i=1}^q (R_i f) \mathcal{L}_X \lambda_i \\
&= d(Xf) - \sum_{i=1}^q (R_i (Xf)) \lambda_i,
\end{align*}
since \( [X, R_i] = 0 \) and \( \mathcal{L}_X \lambda_i = 0 \).

Moreover,
\[
i_{[X, X_f]} \lambda_i = \mathcal{L}_X (i_{X_f} \lambda_i) - i_{X_f} \mathcal{L}_X \lambda_i = 0.
\]
Thus, \( [X, X_f] \) is a Hamiltonian vector field with Hamiltonian function \( X(f) \), as claimed.
\end{proof}

 \begin{remark}
Since every gradient vector field is also a local gradient vector field, gradient vector fields also satisfy conclusion (1) of Theorem~\ref{LGG}.
\end{remark}

\begin{remark}
For an infinitesimal automorphism \( X \) and a \( q \)-evolution vector field \( E_f \) associated to a function \( f \), the vector field
\[
\sum_{i=1}^q R_i + [X, E_f]
\]
is the \( q \)-evolution vector field \( E_{Xf} \) associated to the function \( Xf \).
\end{remark}

\begin{proof}
First, we show that \( [X, R_i] = 0 \) for all \( i = 1, \dots, q \). Indeed, for any \( i,j = 1, \dots, q \), we compute:
\[
0 = \mathcal{L}_X (i_{R_i} \lambda_j) = i_{[X, R_i]} \lambda_j + i_{R_i} \mathcal{L}_X \lambda_j = i_{[X, R_i]} \lambda_j,
\]
which implies \( [X, R_i] \in \xi \). Moreover,
\[
0 = \mathcal{L}_X (i_{R_i} \Omega) = i_{[X, R_i]} \Omega + i_{R_i} \mathcal{L}_X \Omega = i_{[X, R_i]} \Omega.
\]
Since \( \Omega \) is non-degenerate on \( \xi \), we conclude \( [X, R_i] = 0 \) for all \( i \).

Next, by Theorem~\ref{LGG}, we know that \( [X, X_f] \) is the Hamiltonian vector field associated to the function \( X(f) \). Hence, the \( q \)-evolution vector field associated to \( Xf \) is given by
\[
E_{Xf} = \sum_{i=1}^q R_i + X_{X(f)} = \sum_{i=1}^q R_i + [X, X_f].
\]
\end{proof}

We now observe that a \( q \)-cosymplectic manifold is a corank-\( q \) Poisson manifold, with the Poisson bracket defined as follows.

\begin{theorem} \label{D2}
Let \( (M, \Omega, \vec{\lambda}) \) be a \( q \)-cosymplectic manifold. Define the bracket
\[
\{ \cdot, \cdot \} : C^\infty(M) \times C^\infty(M) \to C^\infty(M), \quad (f, g) \mapsto \{f, g\} := \Omega(X_f, X_g),
\]
then this bracket is a Poisson bracket.
\end{theorem}

\begin{proof}
It is clear that the bracket is bilinear and antisymmetric. We now verify the Jacobi identity and the Leibniz rule.

First, we show that
\[
X_{\{f, g\}} = -[X_f, X_g].
\]

For any \( j = 1, \dots, q \), we compute:
\begin{align*}
i_{[R_j, X_f]} \Omega &= \mathcal{L}_{R_j}(i_{X_f} \Omega) 
= \mathcal{L}_{R_j} \left( df - \sum_{i=1}^q R_i(f) \lambda_i \right)
= d(R_j(f)) - \sum_{i=1}^q R_i(R_j(f)) \lambda_i, \\
i_{[R_j, X_f]} \lambda_i &= \mathcal{L}_{R_j} (i_{X_f} \lambda_i) - i_{X_f} \mathcal{L}_{R_j} \lambda_i = 0,
\end{align*}
so we conclude
\begin{align} \label{Ri}
[R_j, X_f] = X_{R_j(f)}, \quad j = 1, \dots, q.
\end{align}

Now, using
\begin{align}
\{f, g\} &= (i_{X_f} \Omega)(X_g) = df(X_g) - \sum_{i=1}^q R_i(f) \lambda_i(X_g) = X_g(f), \label{Xg} \\
\mathcal{L}_{X_f} \lambda_i &= d(i_{X_f} \lambda_i) + i_{X_f} d\lambda_i = 0, \quad i = 1, \dots, q, \label{LL}
\end{align}
we compute:
\begin{align*}
i_{X_{\{f, g\}}} \Omega &= d\{f, g\} - \sum_{i=1}^q R_i(\{f, g\}) \lambda_i 
= d(X_g(f)) - \sum_{i=1}^q R_i(X_g(f)) \lambda_i, \\
i_{-[X_f, X_g]} \Omega &= i_{[X_g, X_f]} \Omega 
= \mathcal{L}_{X_g}(i_{X_f} \Omega) - i_{X_f} \mathcal{L}_{X_g} \Omega \\
&= \mathcal{L}_{X_g} \left( df - \sum_{i=1}^q R_i(f) \lambda_i \right) - i_{X_f} d \left( dg - \sum_{i=1}^q R_i(g) \lambda_i \right) \\
&= d(X_g(f)) - \sum_{i=1}^q X_g(R_i(f)) \lambda_i + \sum_{i=1}^q X_f(R_i(g)) \lambda_i \\
&= d(X_g(f)) - \sum_{i=1}^q X_g(R_i(f)) \lambda_i + \sum_{i=1}^q \{ R_i(g), f \} \lambda_i \\
&= d(X_g(f)) - \sum_{i=1}^q X_g(R_i(f)) \lambda_i - \sum_{i=1}^q X_{R_i(g)}(f) \lambda_i \\
&= d(X_g(f)) - \sum_{i=1}^q [R_i, X_g](f) \lambda_i \\
&= d(X_g(f)) - \sum_{i=1}^q R_i(X_g(f)) \lambda_i.
\end{align*}
Hence,
\begin{align} \label{Xf}
i_{X_{\{f,g\}}} \Omega = i_{-[X_f, X_g]} \Omega.
\end{align}

Also, since
\[
i_{[X_f, X_g]} \lambda_i = \mathcal{L}_{X_f} (i_{X_g} \lambda_i) - i_{X_g} \mathcal{L}_{X_f} \lambda_i = 0, \quad i = 1, \dots, q,
\]
we get
\[
i_{X_{\{f,g\}} + [X_f, X_g]} \lambda_i = 0,
\]
so \( X_{\{f, g\}} = -[X_f, X_g] \).

\medskip

To prove the Jacobi identity, for any \( f, g, h \in C^\infty(M) \), we compute:
\begin{align*}
\{f, \{g, h\}\} &= X_{\{g, h\}}(f) = -[X_g, X_h](f) = X_h(X_g(f)) - X_g(X_h(f)) \\
&= X_h\{f, g\} - X_g\{f, h\} = \{\{f, g\}, h\} - \{\{f, h\}, g\},
\end{align*}
which implies
\[
\{\{f, g\}, h\} + \{\{g, h\}, f\} + \{\{h, f\}, g\} = 0.
\]

\medskip

Finally, we verify the Leibniz rule. For any \( f, g, h \in C^\infty(M) \), using \eqref{Xg}, we get:
\[
\{fg, h\} = X_h(fg) = g X_h(f) + f X_h(g) = g \{f, h\} + f \{g, h\}.
\]

This completes the proof that the bracket defined is  Poisson.
\end{proof}

\begin{theorem}
Let \( (M, \Omega, \vec{\lambda}) \) be a \( q \)-cosymplectic manifold. The relationship between the Hamiltonian vector field and the gradient vector field of a function \( f \in C^\infty(M) \) is given by:
\[
\nabla f = X_f + \sum_{i=1}^q R_i(f) R_i.
\]
Thus,
\[
\{f, g\} = \Omega(X_f, X_g) = \Omega(\nabla f, \nabla g).
\]
\end{theorem}

\begin{proof}
From the definitions of the Hamiltonian and gradient vector fields, we have
\[
i_{\nabla f} \Omega = df - \sum_{i=1}^q R_i(f) \lambda_i = i_{X_f} \Omega \quad \Rightarrow \quad i_{\nabla f - X_f} \Omega = 0.
\]
Therefore, \( \nabla f = X_f + Y \) for some vector field \( Y \) on \( M \) satisfying \( i_Y \Omega = 0 \). Hence,
\[
R_i(f) = i_{\nabla f} \lambda_i = i_{X_f} \lambda_i + i_Y \lambda_i, \quad i = 1, \dots, q.
\]
Since \( i_{X_f} \lambda_i = 0 \) for all \( i \), and because
\[
TM = \left( \bigcap_i \ker \lambda_i \right) \oplus \ker \Omega,
\]
it follows that \( Y = \sum_{i=1}^q R_i(f) R_i \), and so
\[
\nabla f = X_f + \sum_{i=1}^q R_i(f) R_i.
\]
Thus,
\[
\Omega(\nabla f, \nabla g) = \Omega\left(X_f + \sum_{i=1}^q R_i(f) R_i,\; X_g + \sum_{i=1}^q R_i(g) R_i\right) = \Omega(X_f, X_g) = \{f, g\}.
\]
\end{proof}

Consider the flow of the \( q \)-evolution vector field \( E_H = \sum_{i=1}^q R_i + X_H \) on a \( 2n+q \)-dimensional \( q \)-cosymplectic manifold \( (M, \Omega, \vec{\lambda}) \). A function \( f \) is said to be a \emph{first integral} of \( E_H \) if
\[
E_H(f) = \sum_{i=1}^q R_i(f) + X_H(f) = \sum_{i=1}^q R_i(f) + \{f, H\} = 0.
\]
Observe that
\[
E_H(H) = \sum_{i=1}^q R_i(H) + X_H(H) = \sum_{i=1}^q R_i(H) + \{H, H\} = \sum_{i=1}^q R_i(H),
\]
which shows that \( H \) may not be a first integral of \( E_H \).

We now formulate an analogue of Nekhoroshev's theorem on non-commutative integrability (see \cite{Nekhoroshev}).

\begin{theorem} \label{keji}
Consider the flow of the \( q \)-evolution vector field \( E_H = \sum_{i=1}^q R_i + X_H \) on a \( 2n+q \)-dimensional \( q \)-cosymplectic manifold \( (M, \Omega, \vec{\lambda}) \). Assume the following:

\begin{enumerate}
\item[(1)] The vector field \( E_H \) admits \( m \) independent first integrals \( f_1, \dots, f_m \), such that
\[
\{f_i, f_j\} = 0, \quad \text{for all } i = 1, \dots, r, \quad j = 1, \dots, m,\quad \text{with } r < m,
\]
and the condition \( 2n + q - 1 = m + r \) holds.

\item[(2)] The map \( F := (f_{r+1}, \dots, f_m) : M \to \mathbb{R}^{m - r} \) is a submersion whose fibers
\[
M_{\mathbf{c}} = \{x \in M \mid f_{r+1}(x) = c_1, \dots, f_m(x) = c_{m - r}\}, \quad \mathbf{c} \in \mathbb{R}^{m - r},
\]
are compact and connected.
\end{enumerate}

Then each fiber \( M_{\mathbf{c}} \) is diffeomorphic to a torus \( \mathbb{T}^{m - r} \), and the vector field \( E_H \) is conjugate to a constant vector field on \( \mathbb{T}^{m - r} \).
\end{theorem}

\begin{proof}
First, note that the vector fields \( E_H, X_{f_1}, \dots, X_{f_r} \) commute among themselves. The relations
\[
[X_{f_i}, X_{f_j}] = 0, \quad i,j = 1, \dots, r
\]
follow directly from \eqref{Xf} and \eqref{fi}. Furthermore, from \eqref{Ri} we have
\begin{align*}
[E_H, X_{f_i}] &= \left[X_H + \sum_{j=1}^q R_j,\; X_{f_i} \right] = -X_{\{H, f_i\}} + \sum_{j=1}^q X_{R_j(f_i)} \\
&= X_{X_H(f_i)} + X_{\sum_{j=1}^q R_j(f_i)} = X_{E_H(f_i)} = 0, \quad i = 1, \dots, r.
\end{align*}

Next, since \( f_1, \dots, f_m \) are integrals of \( E_H \), the vector field \( E_H \) is tangent to the level sets \( M_{\mathbf{c}} \). Also, since \( X_{f_i}(f_j) = \{f_j, f_i\} = 0 \) for all \( i = 1, \dots, r \), the vector fields \( X_{f_1}, \dots, X_{f_r} \) are tangent to \( M_{\mathbf{c}} \) as well. Therefore, by Theorem \ref{T3}, the result follows.
\end{proof}

Now, given a \( q \)-cosymplectic manifold \( (M, \Omega, \vec{\lambda}) \) and a function \( H' \in C^\infty(M) \) which is a first integral of \( R_i \) for \( i = 1, \dots, q \), we show that another \( q \)-cosymplectic structure can be obtained.

\begin{theorem}
Consider a \( q \)-cosymplectic manifold \( (M, \Omega, \vec{\lambda}) \) and a function \( H' \in C^\infty(M) \) that is a first integral of each Reeb vector field \( R_i, i = 1, \dots, q \). Then the triple \( (M, \Omega', \vec{\lambda}) \), where
\[
\Omega' = \Omega + dH' \wedge \sum_{i=1}^q \lambda_i,
\]
defines another \( q \)-cosymplectic manifold structure on \( M \).
\end{theorem}

\begin{proof}
Define new vector fields \( R_i' := R_i + X_{H'} \). From the definition of the Hamiltonian vector field \( X_{H'} \) and using \eqref{Ri}, we can verify the following:

\begin{align*}
i_{R_i'} \lambda_j &= i_{R_i} \lambda_j + i_{X_{H'}} \lambda_j = i_{R_i} \lambda_j = \delta_i^j, \\[1ex]
[R_i', R_j'] &= [R_i + X_{H'}, R_j + X_{H'}] \\
&= [R_i, R_j] + [X_{H'}, R_j] + [R_i, X_{H'}] + [X_{H'}, X_{H'}] \\
&= -X_{R_j(H')} + X_{R_i(H')} = 0 \quad \text{(since } R_j(H') = R_i(H') = 0), \\[1ex]
i_{R_i'} \Omega' &= i_{R_i + X_{H'}} \left(\Omega + dH' \wedge \sum_{j=1}^q \lambda_j \right) \\
&= i_{R_i} \Omega + i_{X_{H'}} \Omega + i_{R_i} \left(dH' \wedge \sum_{j=1}^q \lambda_j \right) + i_{X_{H'}} \left(dH' \wedge \sum_{j=1}^q \lambda_j \right) \\
&= 0 + i_{X_{H'}} \Omega + R_i(H') \sum_{j=1}^q \lambda_j - dH' = i_{X_{H'}} \Omega - dH' = 0, \\[1ex]
(\Omega')^n \wedge \lambda_1 \wedge \cdots \wedge \lambda_q &= \Omega^n \wedge \lambda_1 \wedge \cdots \wedge \lambda_q \neq 0.
\end{align*}

Hence, \( (M, \Omega', \vec{\lambda}) \) satisfies the conditions for being a \( q \)-cosymplectic manifold.
\end{proof}

Next, we explore the relationship between the Hamiltonian vector fields corresponding to the same function on the two 
\( q \)-cosymplectic manifolds introduced in the previous theorem. This allows us to understand how the Poisson brackets of two functions are related under this transformation.

\begin{theorem}\label{Gou}
Let \( X_f, \{\cdot,\cdot\} \) and \( X_f', \{\cdot,\cdot\}' \) denote the Hamiltonian vector field and the Poisson bracket on \( (M, \Omega, \vec{\lambda}) \) and \( (M, \Omega', \vec{\lambda}) \), respectively. Then \( X_f = X_f' \), and thus the Poisson brackets coincide:
\[
\{f_1, f_2\} = \{f_1, f_2\}'.
\]
\end{theorem}

\begin{proof}
By definition, the Hamiltonian vector fields \( X_f \) and \( X_f' \) satisfy:
\begin{align}
i_{X_f}\Omega &= df - \left( \sum_{i=1}^q R_i(f)\lambda_i \right), \label{1} \\
i_{X_f'}(\Omega + dH' \wedge \sum_{i=1}^q \lambda_i) &= i_{X_f'}\Omega + X_f'(H') \sum_{i=1}^q \lambda_i - \sum_{i=1}^q \lambda_i(X_f') dH' \notag \\
&= df - \sum_{i=1}^q R_i'(f) \lambda_i, \label{2} \\
\lambda_i(X_f) &= \lambda_i(X_f') = 0, \quad i=1,\dots,q. \label{3}
\end{align}

From equation \eqref{2}, we deduce:
\begin{align*}
i_{X_f'}\Omega &= df - \sum_{i=1}^q \left(X_f'(H') + R_i'(f)\right)\lambda_i 
= df - \sum_{i=1}^q \left(-X_{H'}'(f) + R_i(f)\right)\lambda_i 
= df - \sum_{i=1}^q R_i(f)\lambda_i 
= i_{X_f}\Omega.
\end{align*}

Together with equation \eqref{3}, we conclude that \( X_f = X_f' \). Therefore,
\[
\{f_1, f_2\} = X_{f_2}(f_1) = X_{f_2}'(f_1) = \{f_1, f_2\}'.
\]
\end{proof}

For the \( q \)-evolution vector field on the newly constructed \( q \)-cosymplectic manifold, we can now formulate the following integrability theorem.

\begin{theorem}
Let \( (M, \Omega, \vec{\lambda}) \) be a \( 2n + q \)-dimensional \( q \)-cosymplectic manifold. Assume that:

\begin{enumerate}
\item \( H' \in C^\infty(M) \) is a first integral of each Reeb vector field \( R_i, i = 1, \dots, q \).

\item The \( q \)-evolution vector field \( E_H = \sum_{i=1}^q R_i + X_H \) has \( m \) independent integrals \( f_1, \dots, f_m \), such that
\begin{align}\label{fi}
\{H', f_k\} = 0, \quad k = 1, \dots, q, \qquad
\{f_i, f_j\} = 0, \quad i = 1, \dots, r,\quad j = 1, \dots, m, \quad r < m,
\end{align}
and \( 2n + q - 1 = m + r \).

\item The map \( F := (H, f_{r+1}, \dots, f_m): M \rightarrow \mathbb{R}^{m - r + 1} \) is a submersion with compact and connected fibers:
\[
M_{\mathbf{c}} = \{ x \in M \;|\; H(x) = c_1,\; f_{r+1}(x) = c_2, \dots, f_m(x) = c_{m - r + 1} \}, \quad \mathbf{c} = (c_1, \dots, c_{m - r + 1}) \in \mathbb{R}^{m - r + 1}.
\]
\end{enumerate}

Then, each fiber \( M_{\mathbf{c}} \) is diffeomorphic to \( \mathbb{T}^{m - r + 1} \), and the \( q \)-evolution vector field \( E_H' := \sum_{i=1}^q R_i' + X_H' \), defined on the \( q \)-cosymplectic manifold \( (M, \Omega', \vec{\lambda}) \), is conjugate to a constant vector field on \( \mathbb{T}^{m - r + 1} \).
\end{theorem}

\begin{proof}
Similar to the proof of Theorem~\ref{keji}, and using the conclusion of Theorem~\ref{Gou}, we compute:
\begin{align*}
[E_H', X'_{f_i}] &= \left[ X'_H + \sum_{j=1}^q R'_j, X'_{f_i} \right] = -X'_{\{H, f_i\}'} + \sum_{j=1}^q X'_{R'_j(f_i)} \\
&= -X'_{\{H, f_i\}} + X'_{\sum_{j=1}^q (R_j + X_{H'})(f_i)} \\
&= X'_{X_H(f_i)} + X'_{\sum_{j=1}^q (R_j(f_i) + \{f_i, H'\})} \\
&= X'_{X_H(f_i)} + X'_{\sum_{j=1}^q R_j(f_i)} \\
&= X'_{E_H(f_i)} = 0, \quad i = 1, \dots, r.
\end{align*}

Moreover, we have:
\begin{align*}
[X'_{f_i}, X'_{f_j}] &= -\{f_i, f_j\}' = -\{f_i, f_j\} = 0, \quad i,j = 1, \dots, r, \\
E'_H(f_i) &= \sum_{j=1}^q R'_j(f_i) + X_H'(f_i) 
= \sum_{j=1}^q (R_j + X_{H'})(f_i) + X_H(f_i) \\
&= \sum_{j=1}^q R_j(f_i) + X_H(f_i) = E_H(f_i) = 0.
\end{align*}

Thus, \( f_1, \dots, f_m \) are integrals of \( E_H' \), and \( E_H' \) is tangent to \( M_{\mathbf{c}} \). Moreover, since
\[
X'_{f_i}(f_j) = \{f_j, f_i\}' = \{f_j, f_i\} = 0, \quad i = 1, \dots, r,
\]
we have that \( X'_{f_1}, \dots, X'_{f_r} \) are also tangent to \( M_{\mathbf{c}} \). Then, by Theorem~\ref{T3}, the result follows.
\end{proof}

In the following theorem, we describe how symplectic and \( q \)-cosymplectic manifolds are naturally related.

\begin{theorem}\label{Con}
Let \( M \) be a \( 2n + q \)-dimensional manifold, \( \omega \in \Omega^2(M) \), and \( \eta_i \in \Omega^1(M) \), \( i = 1, \dots, q \). Let \( \operatorname{pr} : \mathbb{R}^q \times M \to M \) be the canonical projection onto \( M \), and let \( (s_1, \dots, s_q) \) be the natural coordinates on \( \mathbb{R}^q \), understood as variables on \( \mathbb{R}^q \times M \) in the standard way. Then, \( (M, \omega, \vec{\eta} = (\eta_1, \dots, \eta_q)) \) is a \( q \)-cosymplectic manifold if and only if \( (\mathbb{R}^q \times M, \widehat{\omega}) \) is a symplectic manifold, where
\[
\widehat{\omega} := \operatorname{pr}^* \omega + \sum_{i=1}^q ds_i \wedge \operatorname{pr}^* \eta_i.
\]
Moreover, \( \operatorname{pr} \) is a Poisson morphism, i.e.,
\[
\{f \circ \operatorname{pr}, k \circ \operatorname{pr} \}_{\widehat{\omega}} = \{f, k\}_{\omega, \vec{\eta}} \circ \operatorname{pr}, \quad \forall f, k \in C^\infty(M),
\]
where \( \{\cdot,\cdot\}_{\widehat{\omega}} \) is the Poisson bracket on the symplectic manifold \( (\mathbb{R}^q \times M, \widehat{\omega}) \), and \( \{\cdot,\cdot\}_{\omega, \vec{\eta}} \) is the Poisson bracket on the \( q \)-cosymplectic manifold \( (M, \omega, \vec{\eta}) \).
\end{theorem}

\begin{proof}
Since \( (M, \omega, \vec{\eta}) \) is a \( q \)-cosymplectic manifold and \( \dim M = 2n + q \), the form \( \omega^n \wedge \eta_1 \wedge \cdots \wedge \eta_q \) is a volume form. Also, \( \omega \in \Omega^2(M) \) and each \( \eta_i \in \Omega^1(M) \) is closed, hence
\begin{align}\label{dome}
d\widehat{\omega} = d\left( \operatorname{pr}^* \omega + \sum_{i=1}^q ds_i \wedge \operatorname{pr}^* \eta_i \right) 
= \operatorname{pr}^* d\omega - \sum_{i=1}^q ds_i \wedge \operatorname{pr}^* d\eta_i = 0,
\end{align}
so \( \widehat{\omega} \) is closed.

Since for any differential 2-form \( \theta \) on a manifold, \( \theta^{n+q} = 0 \) unless it is non-degenerate of rank \( 2(n+q) \), we compute:
\begin{align}\label{ome}
\widehat{\omega}^{n+q} 
&= \left( \operatorname{pr}^* \omega + \sum_{i=1}^q ds_i \wedge \operatorname{pr}^* \eta_i \right)^{n+q} \nonumber \\
&= \binom{n+q}{n} \cdot q! \cdot (\operatorname{pr}^* \omega)^n \wedge ds_1 \wedge \operatorname{pr}^* \eta_1 \wedge \cdots \wedge ds_q \wedge \operatorname{pr}^* \eta_q \nonumber \\
&= \binom{n+q}{n} \cdot q! \cdot (-1)^{\frac{q(q-1)}{2}} \cdot ds_1 \wedge \cdots \wedge ds_q \wedge \operatorname{pr}^* (\omega^n \wedge \eta_1 \wedge \cdots \wedge \eta_q),
\end{align}
which is clearly a volume form on \( \mathbb{R}^q \times M \). Thus, \( \widehat{\omega} \) is non-degenerate, and \( (\mathbb{R}^q \times M, \widehat{\omega}) \) is a symplectic manifold.

	Conversely, if \( (\mathbb{R}^q \times M, \widehat{\omega}) \) is a symplectic manifold, then equation~\eqref{ome} implies that
\[
\omega^n \wedge \eta_1 \wedge \cdots \wedge \eta_q \neq 0.
\]
Moreover, from equation~\eqref{dome}, we have that \( \operatorname{pr}^*\omega \) and \( \operatorname{pr}^* \eta_i \), \( i = 1, \dots, q \), are closed forms. Since \( \operatorname{pr} \) is a surjective submersion, it follows that \( d\omega = 0 \) and \( d\eta_i = 0 \), for all \( i = 1, \dots, q \). Therefore, \( (M, \omega, \vec{\eta} = (\eta_1, \dots, \eta_q)) \) is a \( q \)-cosymplectic manifold.

Moreover, since \( \{\cdot, \cdot\}_{\widehat{\omega}} \) is the Poisson bracket induced by the symplectic form \( \widehat{\omega} \), we have
\begin{align*}
\operatorname{pr}^* \{f, k\}_{\omega, \vec{\eta}} 
&= \operatorname{pr}^*\left( i_{X_k} i_{X_f} \omega \right)
= i_{X_{\operatorname{pr}^*k}} i_{X_{\operatorname{pr}^*f}} \left( \operatorname{pr}^* \omega \right) \\
&= i_{X_{\operatorname{pr}^*k}} d(\operatorname{pr}^* f) 
= \{ \operatorname{pr}^* f, \operatorname{pr}^* k \}_{\widehat{\omega}},
\end{align*}
where \( X_{\operatorname{pr}^* k} \) denotes the Hamiltonian vector field on \( (\mathbb{R}^q \times M, \widehat{\omega}) \) associated to the function \( \operatorname{pr}^* k \in C^\infty(\mathbb{R}^q \times M) \). 

This proves that \( \operatorname{pr} \) is a Poisson morphism.

\end{proof}
As is well known, on symplectic manifold $(M,\omega)$, Hamiltonian systems are equations of the form $\dot z=X_f(z),$ where $f$ is a smooth function on $M$ and $X_f$ is a special vector field on $M$ called the Hamiltonian vector field (defined by the formula $i_{X_f}\omega=df$). 
\begin{corollary}
	If $(M,\omega,\vec{\eta}=(\eta_1,...,\eta_q))$ described as in Theorem \ref{Con} is a $q$-cosymplectic manifold, then natural extension to $\mathbb R^q\times M$ of its Reeb vector field $R_i,i=1,...,q$ are Hamiltonian vector fields with respect to $\hat\omega=pr^{*}\omega+\sum_{i=1}^qds_i\wedge pr^*\eta_i.$
\end{corollary}
\begin{proof}
	Here, we still use $R_i,i=1,...,q$ to denote the natural extension of Reeb vector fields.   Now we can calculate that 
	$$i_{R_i}\hat\omega=i_{R_i}\left(pr^{*}\omega+\sum_{i=1}^qds_i\wedge pr^*\eta_i\right)=ds_i,\quad i=1,...,q,$$
	thus $R_i,i=1,...,q$ are Hamiltonian vector fields with Hamiltonian function $s_i,i=1,...,q.$
\end{proof}

\section{Marsden–Weinstein reduction theorem On $q$-Cosymplectic Manifold}
In this section, we will deal with smooth action $\Phi:G\times M\rightarrow M$ of a Lie group $G$ on a $q$-cosymplectic manifold $(M,\Omega,\vec{\lambda}).$ It will always be tacitly assumed that both $G$ and $M$ are connected. The Lie algebra of $G$ will be denoted by $\mathfrak{g}$ and its dual by $\mathfrak{g}^*.$ For each $g\in G,$ we put $\Phi_g=\Phi(g,\cdot),$ the induced transformation on $M.$ The fundamental vector field, or infinitesimal generator, associated with $\xi\in\mathfrak{g}$ is the vector field $\xi_M$ on $M$ defined by
\begin{align*}
	\xi_M(x)=\frac{d}{dt}\Phi(\exp(t\xi),x)\big|_{t=0}.
\end{align*}
An action $\Phi$ of a Lie group $G$ on a $q$-cosymplectic manifold $(M,\Omega,\vec{\lambda})$ is called an \emph{automorphism action}, or $G$ is said to be an automorphism group of $(M,\Omega,\vec{\lambda})$, if for each $g\in G$ the corresponding $\Phi_g$ is an automorphism of the $q$-cosymplectic structure, i.e.,
\[
\Phi^*_g\Omega=\Omega,\quad \Phi^*_g\lambda_i=\lambda_i,\quad i=1,...,q.
\]
It then follows that the fundamental vector fields are infinitesimal automorphisms, i.e.,
\[
L_{\xi_M}\Omega=0,\quad L_{\xi_M}\lambda_i=0,\quad i=1,...,q,
\]
for each $\xi\in\mathfrak{g}.$ An automorphism action $\Phi$ will be called a \emph{restricted Hamiltonian action} if furthermore, for each $\xi\in \mathfrak{g}$, the associated fundamental vector field $\xi_M$ is a Hamiltonian vector field with Hamiltonian function $J_\xi$ and  $\xi_M\in\xi,$ i.e.,
\begin{align}\label{Lie}
	i_{\xi_M}\Omega=dJ_\xi,\quad i_{\xi_M}\lambda_i=0,\quad i=1,...,q.
\end{align}
With the restricted Hamiltonian action, one can associate a momentum map $J^\Phi:M\rightarrow\mathfrak{g}^*$ defined in the usual way by 
$$<\xi,J^\Phi(x)>=J_\xi(x)$$
for all $\xi\in\mathfrak{g}$.
\begin{remark}
It is clear that  $R_j(J_\xi)=0,j=1,...,q,$ since
$$0=i_{R_j}i_{\xi_M}\Omega=i_{R_j}(dJ_\xi)=R_j(J_\xi),\quad j=1,...,q.$$
\end{remark}
\begin{remark}
	It is easy to see that if $J_\xi$ is a Hamiltonian function of $\xi_M,$ then $J_\xi+c$ is also a Hamiltonian function of $\xi_M,$ where $c$ is a constant. Thus, the non-uniqueness of the Hamiltonian $J_\xi$ is reflected in the non-uniqueness of the momentum map.
\end{remark}

\begin{definition}
	A momentum map $J^\Phi:M\rightarrow\mathfrak{g}^*$ is Ad$^*$-equivariant if 
	\[
	J^\Phi\circ\Phi_g=Ad^*_{g^{-1}}\circ J^\Phi,\quad \forall g\in G.
	\]
	In other words, the following diagram commutes:
	\[
	\begin{tikzcd}
		M \arrow[r, "J^\Phi"] \arrow[d, "\Phi_g"'] & \mathfrak{g}^* \arrow[d, "Ad^*_{g^{-1}}"] \\
		M \arrow[r, "J^\Phi"] & \mathfrak{g}^*
	\end{tikzcd}
	\]
	for every $g\in G$, where $Ad^*_{g^{-1}}$ denotes the dual (transpose) of $Ad_{g^{-1}}$.
\end{definition}

To simplify the notation, the $q$-cosymplectic manifold $(M,\Omega,\vec{\lambda})$ will be frequently denoted by $M_{\vec{\lambda}}^\Omega$.

\begin{definition}
	A triple $(M_{\vec{\lambda}}^\Omega,f,J^\Phi)$ is called a \emph{$G$-invariant $q$-cosymplectic Hamiltonian system} if it consists of a $q$-cosymplectic manifold $(M,\Omega,\vec{\lambda})$, an associated restricted Hamiltonian action $\Phi:G\times M\rightarrow M$ such that $\Phi_g^*f = f$ for every $g\in G$, and a momentum map $J^\Phi$ related to $\Phi$.
\end{definition}

As we know, not every momentum map is Ad$^*$-equivariant. We now develop the theory of non-Ad$^*$-equivariant momentum maps, similar to the theory developed by Alber \cite{Albert} for cosymplectic manifolds. In this paper, we will always assume that the manifold $M$ is connected unless otherwise stated.

\begin{proposition}
	Let $(M_{\vec{\lambda}}^\Omega,f,J^\Phi)$ be a $G$-invariant $q$-cosymplectic Hamiltonian system and define the functions on $M$ of the form
	\[
		\psi_{g,\xi}(x) := J_{\xi}(\Phi_g(x)) - J_{Ad_{g^{-1}}\xi}(x), \quad \forall g\in G,\;\forall \xi\in\mathfrak{g}.
	\]
	Then, $\psi_{g,\xi}$ is constant on $M$ for every $g\in G$ and $\xi\in\mathfrak{g}$. Moreover, the map $\sigma:G\rightarrow \mathfrak{g}^*$ defined by
	\[
		\langle \sigma(g), \xi \rangle := \psi_{g,\xi}
	\]
	satisfies the cocycle condition
	\begin{align}\label{gg}
		\sigma(gg') = \sigma(g) + Ad^*_{g^{-1}}\sigma(g'),\quad \forall g,g'\in G.
	\end{align}
\end{proposition}

\begin{proof}
	In fact, using the relations $(Ad_g\xi)_M = \Phi_{g*}\xi_M$ and $\Phi_g^*\Omega = \Omega$, we compute for all $g\in G$, $\xi\in\mathfrak{g}$:
	\begin{align*}
		d\psi_{g,\xi} &= d(J_\xi\circ\Phi_g) - dJ_{Ad_{g^{-1}}\xi} 
		= \Phi_g^* dJ_\xi - i_{(Ad_{g^{-1}}\xi)_M} \Omega \\
		&= \Phi_g^*(i_{\xi_M}\Omega) - i_{\Phi_{g^{-1}*}\xi_M}\Omega 
		= i_{\Phi_{g^{-1}*}\xi_M}\Omega - i_{\Phi_{g^{-1}*}\xi_M}\Omega = 0.
	\end{align*}
	Hence, $\psi_{g,\xi}$ is constant on $M$ for all $g\in G$ and $\xi\in\mathfrak{g}$.

	To express all $\psi_{g,\xi}$ simultaneously, observe:
	\begin{align*}
		\psi_{g,\xi}(x) &= J_{\xi}(\Phi_g(x)) - J_{Ad_{g^{-1}}\xi}(x) 
		= \langle J^\Phi(\Phi_g(x)), \xi \rangle - \langle J^\Phi(x), Ad_{g^{-1}}\xi \rangle \\
		&= \langle J^\Phi(\Phi_g(x)) - Ad_{g^{-1}}^*J^\Phi(x), \xi \rangle.
	\end{align*}
	Thus, we define the map:
	\begin{align}\label{sigma}
		\sigma(g) := J^\Phi \circ \Phi_g - Ad^*_{g^{-1}} \circ J^\Phi,
	\end{align}
	which is constant on $M$ for each $g\in G$.

	Now compute:
	\begin{align*}
		\sigma(gg') &= J^\Phi\circ\Phi_{gg'} - Ad^*_{(gg')^{-1}} \circ J^\Phi 
		= J^\Phi\circ\Phi_g\circ\Phi_{g'} - Ad^*_{g^{-1}} Ad^*_{g'^{-1}} J^\Phi \\
		&= \big(J^\Phi\circ\Phi_g - Ad^*_{g^{-1}}J^\Phi\big) + Ad^*_{g^{-1}}\big(J^\Phi\circ\Phi_{g'} - Ad^*_{g'^{-1}}J^\Phi\big) \\
		&= \sigma(g) + Ad^*_{g^{-1}} \sigma(g'),
	\end{align*}
	which proves the cocycle condition \eqref{gg}.
\end{proof}

 As in the symplectic case, the map \eqref{sigma} is called the \emph{co-adjoint cocycle} associated with the $q$-cosymplectic momentum map $J^\Phi$ on $M$. Again as in the symplectic case, $J^\Phi$ is an Ad$^*$-equivariant momentum map if and only if $\sigma = 0$. Roughly speaking, $\sigma$ measures the \emph{lack} of Ad$^*$-equivariance of a $q$-cosymplectic momentum map.

A map $\sigma:G \to \mathfrak{g}^*$ is a \emph{coboundary} if there exists $\mu \in \mathfrak{g}^*$ such that
\begin{align}
	\sigma(g) = \mu - Ad^*_{g^{-1}}\mu, \quad \forall g \in G.
\end{align}
Every coboundary satisfies the cocycle condition \eqref{gg}, since
\begin{align*}
	\sigma(gg') &= \mu - Ad^*_{(gg')^{-1}}\mu = \mu - Ad^*_{g^{-1}} Ad^*_{g'^{-1}} \mu \\
	&= \mu - Ad^*_{g^{-1}} \mu + Ad^*_{g^{-1}} \mu - Ad^*_{g^{-1}} Ad^*_{g'^{-1}} \mu \\
	&= \sigma(g) + Ad^*_{g^{-1}} \sigma(g').
\end{align*}

The space of co-adjoint cocycles admits an equivalence relation, whose equivalence classes are called \emph{cohomology classes}, defined by declaring that two co-adjoint cocycles belong to the same class if their difference is a coboundary. The following proposition shows that every $q$-cosymplectic action admitting a $q$-cosymplectic momentum map induces a well-defined cohomology class $[\sigma]$. Note that these results could be straightforwardly adapted to the symplectic case.

\begin{proposition}\label{PP}
	Let $\Phi: G \times M \to M$ be a $q$-cosymplectic Lie group action relative to $(M_{\vec{\lambda}}^\Omega, f, J^\Phi)$. If $J_1^\Phi$ and $J_2^\Phi$ are two momentum maps associated with $\Phi$, with co-adjoint cocycles $\sigma_1$ and $\sigma_2$ respectively, then $[\sigma_1] = [\sigma_2]$.
\end{proposition}

\begin{proof}
	From the definition of a co-adjoint cocycle for a $q$-cosymplectic momentum map,
	\[
		\langle \sigma_1(g) - \sigma_2(g), \xi \rangle 
		= \langle J_1^\Phi \circ \Phi_g - J_2^\Phi \circ \Phi_g, \xi \rangle 
		- \langle Ad^*_{g^{-1}} (J_1^\Phi - J_2^\Phi), \xi \rangle,
	\]
	for all $g \in G$, $\xi \in \mathfrak{g}$. However, $J_1^\Phi - J_2^\Phi$ takes a constant value $\mu \in \mathfrak{g}^*$ since both maps are momentum maps for the same $q$-cosymplectic action. Indeed,
	\[
		d \langle J_1^\Phi - J_2^\Phi, \xi \rangle = dJ_{1,\xi} - dJ_{2,\xi} 
		= i_{\xi_M} \Omega - i_{\xi_M} \Omega = 0, \quad \forall \xi \in \mathfrak{g}.
	\]
	Thus, $(J_1^\Phi - J_2^\Phi) \circ \Phi_g = J_1^\Phi - J_2^\Phi$ for every $g \in G$, and hence
	\[
		\sigma_1(g) - \sigma_2(g) = \mu - Ad^*_{g^{-1}} \mu, \quad \forall g \in G,
	\]
	where $\mu = J_1^\Phi - J_2^\Phi$, showing that the cocycles differ by a coboundary.
\end{proof}

  Proposition \ref{PP} yields that a $q$-cosymplectic Lie group action has an Ad$^*$-equivariant momentum map if and only if it has an associated coboundary. Indeed, if a $q$-cosymplectic Lie group action has an Ad$^*$-equivariant momentum map $J_2^\Phi$ relative to $(M,\Omega,\vec{\lambda})$, then its associated co-adjoint cocycle satisfies $\sigma_2 = 0$, and any other momentum map $J_1^\Phi$ for the same action is such that its co-adjoint cocycle, say $\sigma_1$, satisfies $[\sigma_1] = [\sigma_2] = 0$, and thus $\sigma_1$ becomes a coboundary. 

Moreover, if $\sigma_1$ is a coboundary induced by $\mu \in \mathfrak{g}^*$, then the momentum map
\[
	J^\Phi := J_1^\Phi - \mu
\]
is an Ad$^*$-equivariant momentum map for the same $q$-cosymplectic Lie group action as $J_1^\Phi$, where $\mu \in \mathfrak{g}^*$ satisfies $\sigma_1(g) = \mu - Ad^*_{g^{-1}}\mu$ for every $g \in G$. In fact,
\[
	\langle J^\Phi, \xi \rangle = \langle J_1^\Phi, \xi \rangle - \langle \mu, \xi \rangle = J_{1,\xi} - \langle \mu, \xi \rangle, \quad \forall \xi \in \mathfrak{g},
\]
and
\[
	\sigma(g) = J^\Phi \circ \Phi_g - Ad^*_{g^{-1}} J^\Phi = \sigma_1(g) + Ad^*_{g^{-1}} \mu - \mu = 0,
\]
for every $g \in G$. Therefore, the result follows.

To summarise, if the co-adjoint cocycle of a given momentum map is a coboundary, then we can construct an Ad$^*$-equivariant momentum map. However, the following proposition shows that for every momentum map there exists a Lie group action 
\[
	\Delta : G \times \mathfrak{g}^* \to \mathfrak{g}^*
\]
such that the momentum map becomes $\Delta$-equivariant; namely, for every $g \in G$ the following diagram commutes:
\[
\begin{tikzcd}
	M \arrow[r, "\Phi_g"] \arrow[d, "J^{\Phi}"'] & M \arrow[d, "J^\Phi"] \\
	\mathfrak{g}^* \arrow[r, "\Delta_g"] & \mathfrak{g}^*
\end{tikzcd}
\]

\begin{proposition}\label{PJ}
Let $J^\Phi : M \to \mathfrak{g}^*$ be a momentum map for a $q$-cosymplectic restricted Hamiltonian action $\Phi : G \times M \to M$ with co-adjoint cocycle $\sigma$. Then:

\begin{enumerate}
	\item The map $\Delta : (g, \mu) \in G \times \mathfrak{g}^* \mapsto \Delta_g(\mu) := Ad^*_{g^{-1}} \mu + \sigma(g) \in \mathfrak{g}^*$ is a Lie group action.

	\item The momentum map $J^\Phi$ is $\Delta$-equivariant.
\end{enumerate}
\end{proposition}

\begin{proof}
First, since $\sigma(e) = 0$, one has
\[
	\Delta(e, \mu) = Ad^*_{e^{-1}} \mu + \sigma(e) = \mu.
\]
Moreover,
\begin{align*}
	\Delta(g, \Delta(g', \mu)) 
	&= Ad^*_{g^{-1}} (Ad^*_{g'^{-1}} \mu + \sigma(g')) + \sigma(g) \\
	&= Ad^*_{g^{-1}} Ad^*_{g'^{-1}} \mu + Ad^*_{g^{-1}} \sigma(g') + \sigma(g) \\
	&= Ad^*_{(gg')^{-1}} \mu + \sigma(gg') = \Delta(gg', \mu).
\end{align*}
Thus, $\Delta$ is a Lie group action on $\mathfrak{g}^*$, which proves part (1).

For part (2), from the definition of $\Delta$ and $\sigma$, we obtain
\[
	\Delta_g \circ J^\Phi = Ad^*_{g^{-1}} J^\Phi + \sigma(g) = J^\Phi \circ \Phi_g, \quad \forall g \in G,
\]
which shows that $J^\Phi$ is $\Delta$-equivariant.
\end{proof}

  Proposition \ref{PJ} ensures that a general $q$-cosymplectic momentum map $J^\Phi$ gives rise to an equivariant momentum map relative to a new action $\Delta:G\times\mathfrak{g}^*\rightarrow\mathfrak{g}^*$, called an \emph{affine action}. 

In the following theorem, we analyse the commutation relations among the functions $\{J_\xi\}_{\xi\in\mathfrak{g}}$ associated with a $q$-cosymplectic momentum map $J^\Phi$.

\begin{theorem}\label{Sigma}
Let $\Phi:G\times M\rightarrow M$ be a $q$-cosymplectic restricted Hamiltonian action relative to $(M,\Omega,\vec{\lambda})$ with a $q$-cosymplectic momentum map $J^\Phi:M\rightarrow\mathfrak{g}^*$ and let $\sigma:G\rightarrow\mathfrak{g}^*$ be the co-adjoint cocycle of $J^\Phi.$ Let us define
\[
\sigma_\tau: g \in G \mapsto \left<\sigma(g),\tau\right> \in \mathbb{R}, \quad
\Sigma : (\xi_1,\xi_2) \in \mathfrak{g} \times \mathfrak{g} \mapsto T_e \sigma_{\xi_2}(\xi_1) \in \mathbb{R}, \quad \forall \tau \in \mathfrak{g}.
\]
Then:
\begin{enumerate}
    \item The map $\Sigma$ is a skew-symmetric bilinear form on $\mathfrak{g}$ satisfying the identity
    \[
    \Sigma(\xi,[\zeta,\nu]) + \Sigma(\nu,[\xi,\zeta]) + \Sigma(\zeta,[\nu,\xi]) = 0, \quad \forall \xi,\zeta,\nu \in \mathfrak{g}.
    \]
    
    \item $\Sigma(\xi,\nu) = \{J_\nu, J_\xi\}_{\Omega,\vec{\lambda}} - J_{[\nu,\xi]}$ does not depend on $x \in M$, for all $\xi,\nu \in \mathfrak{g}.$
\end{enumerate}
\end{theorem}

\begin{proof}
Let us prove part 2 first. Taking the tangent map of $\sigma_\tau$ at $e$, we get 
\begin{align*}
\Sigma(\xi,\tau) &= T_e \sigma_\tau(\xi) = \frac{d}{ds}\bigg|_{s=0} \left( \left< J^\Phi(\Phi_{\exp(s\xi)}), \tau \right> - \left< Ad^*_{\exp(-s\xi)} J^\Phi(x), \tau \right> \right) \\
&= dJ_\tau(\xi_M)_x - \frac{d}{ds}\bigg|_{s=0} \left< J^\Phi(x), Ad_{\exp(-s\xi)} \tau \right> \\
&= - (i_{\tau_M} i_{\xi_M} \Omega)_x - \left< J^\Phi(x), [\tau, \xi] \right> = \{ J_\tau, J_\xi \}_{\Omega, \vec{\lambda}}(x) - J_{[\tau,\xi]}(x).
\end{align*}
Since $X_{\{J_\tau,J_\xi\}_{\Omega,\vec{\lambda}}} = -[X_{J_\tau}, X_{J_\xi}] = -[\tau_M,\xi_M] = [\tau,\xi]_M$, the vector fields $X_{\{J_\tau,J_\xi\}}$ and $X_{J_{[\tau,\xi]}}$ have the same Hamiltonian function up to a constant. Therefore, $\Sigma$ does not depend on $x \in M$, which proves part 2.

For part 1, we note:
\[
-\Sigma(\xi,[\zeta,\nu]) = \{J_\xi,J_{[\zeta,\nu]}\}_{\Omega,\vec{\lambda}} - J_{[\xi,[\zeta,\nu]]} = \{ J_\xi, \{ J_\zeta, J_\nu \}_{\Omega,\vec{\lambda}} - \Sigma(\nu,\zeta) \}_{\Omega,\vec{\lambda}} - J_{[\xi,[\zeta,\nu]]}
\]
and the result follows from the Jacobi identity and antisymmetry of both the bracket $\{\cdot,\cdot\}_{\Omega,\vec{\lambda}}$ and the Lie bracket.
\end{proof}

Recall that for an Ad$^*$-equivariant momentum map, $\sigma(g) = 0$ for every $g \in G$. Thus, $\Sigma(\xi, \tau) = 0$ for all $\xi, \tau \in \mathfrak{g}$ if $J^\Phi$ is Ad$^*$-equivariant.

By part 2 of Theorem \ref{Sigma}, we know that if $J^\Phi$ is an Ad$^*$-equivariant $q$-cosymplectic momentum map, then there exists a Lie algebra morphism $\xi \in \mathfrak{g} \mapsto J_\xi \in C^\infty(M)$.

The following proposition shows that the momentum map $J^\Phi:M\rightarrow \mathfrak{g}^*$ related to $(M^\Omega_{\vec{\lambda}},f,J^\Phi)$ is conserved for the dynamics of the vector fields $\nabla f$, $X_f$ and $E_f$. In other words, the flows of $\nabla f$, $X_f$ and $E_f$ leave $J^\Phi$ invariant for every $f \in C^\infty(M)$.

\begin{proposition}\label{PXH}
Let $(M^\Omega_{\vec{\lambda}},f,J^\Phi)$ be a $G$-invariant $q$-cosymplectic Hamiltonian system, and let $F:(s,m)\in\mathbb{R} \times M \mapsto F_s(m)$ be the flow of $\nabla f$. Then, $J^\Phi \circ F_s = J^\Phi$ for every $s \in \mathbb{R}$. Analogous results apply to the flows of $E_f$ and $X_f$.
\end{proposition}

\begin{proof}
Since $f$ is $G$-invariant, $\xi_M f = 0$ for every $\xi \in \mathfrak{g}$. Therefore,
\begin{align*}
\frac{d}{ds}\bigg|_{s=0} J_\xi \circ F_s &= i_{\nabla f} dJ_\xi = i_{X_f + \sum_{i=1}^q (R_i f) R_i} dJ_\xi = i_{X_f} dJ_\xi \\
&= i_{X_f} i_{\xi_M} \Omega = i_{\xi_M} \left( \sum_{i=1}^q (R_i f) \lambda_i - df \right) = 0,
\end{align*}
for every $\xi \in \mathfrak{g}$. Hence $J_\xi \circ F_s = J_\xi$ for all $s \in \mathbb{R}$.

Similarly, if $L$ is the flow of $E_f$, then
\begin{align*}
\frac{d}{ds}\bigg|_{s=0} J_\xi \circ L_s &= i_{E_f} dJ_\xi = i_{X_f + \sum_{i=1}^q R_i} dJ_\xi = i_{X_f} dJ_\xi \\
&= i_{X_f} i_{\xi_M} \Omega = i_{\xi_M} \left( \sum_{i=1}^q (R_i f) \lambda_i - df \right) = 0,
\end{align*}
so $J^\Phi \circ L_s = J^\Phi$. Likewise, since $i_{X_f + \sum_{i=1}^q R_i} dJ_\xi = i_{X_f} dJ_\xi$, it follows that $J^\Phi \circ K_s = J^\Phi$ for the flow $K_s$ of $X_f$.
\end{proof}

The following lemma is a generalisation of a well-known result in symplectic geometry, which is crucial to obtain the $q$-cosymplectic Marsden–Weinstein reduction theorem. Recall that a \emph{weakly regular value} of $J^\Phi : M \to \mathfrak{g}^*$ is a point $\mu \in \mathfrak{g}^*$ such that $J^{\Phi-1}(\mu)$ is a submanifold of $M$ and $T_x J^{\Phi-1}(\mu) = \ker T_x J^\Phi$ for every $x \in J^{\Phi-1}(\mu)$ (see \cite{Albert}).

\begin{lemma}
Let $\mu \in \mathfrak{g}^*$ be a weak regular value of a $q$-cosymplectic momentum map $J^\Phi : M \to \mathfrak{g}^*$, and let $G^\Delta_\mu:=\{g\in G|\Delta_g(\mu)=\mu\}$ be the isotropy group at $\mu$ of the affine action $\Delta : G \times \mathfrak{g}^* \to \mathfrak{g}^*$ associated to the co-adjoint cocycle $\sigma : G \to \mathfrak{g}^*$ of $J^\Phi$. Then, for every $x \in J^{\Phi-1}(\mu)$:

\begin{enumerate}
    \item $T_x(G^\Delta_\mu \cdot x) = T_x(G \cdot x) \cap T_x(J^{\Phi-1}(\mu))$,
    \item $T_x(J^{\Phi-1}(\mu)) = T_x(G \cdot x)^{\perp_\Omega}$,
    \item $(T_x(J^{\Phi-1}(\mu)))^{\perp_\Omega} = T_x(G \cdot x) \oplus \langle R_{1x},\dots,R_{qx} \rangle$.
\end{enumerate}
\end{lemma}

\begin{proof}
Let us assume $(\xi_M)_x \in T_x J^{\Phi-1}(\mu)$. Since $T_x(J^{\Phi-1}(\mu)) = \ker T_x J^\Phi$, we compute
\begin{align*}
(i_{\xi_M} dJ_\tau)_x &= \frac{d}{du}\bigg|_{u=0} J_\tau(\Phi(\exp(u\xi), x)) \\
&= \left< \frac{d}{du}\bigg|_{u=0} J^\Phi(\Phi(\exp(u\xi),x)), \tau \right> = \left< \frac{d}{du}\bigg|_{u=0} \Delta_{\exp(u\xi)} J^\Phi(x), \tau \right> = 0,
\end{align*}
for every $\tau \in \mathfrak{g}$, if and only if $\xi \in \mathfrak{g}^\Delta_\mu$, the Lie algebra of $G^\Delta_\mu$. This proves 1.

For part 2, since $J^\Phi$ is a $q$-cosymplectic momentum map, we have
\[
\Omega_x((\xi_M)_x, v_x) = (dJ_\xi)_x(v_x) = \left< T_x J^\Phi(v_x), \xi \right>, \quad \forall v_x \in T_x M, \, \forall \xi \in \mathfrak{g}.
\]
Thus, $v_x \in \ker T_x J^\Phi = T_x(J^{\Phi-1}(\mu))$ if and only if $T_x J^\Phi(v_x)$ is orthogonal to $\mathfrak{g}$, i.e., $T_x(J^{\Phi-1}(\mu)) = T_x(Gx)^{\perp_\Omega}$. This proves 2.

For part 3, take $X = \xi_M + \sum_{i=1}^q \lambda_i R_i$ for any $\lambda_i \in \mathbb{R}$. Then, for $v_x \in \ker T_x J^\Phi$,
\[
\Omega_x(X_x, v_x) = (dJ_\xi)_x(v_x) = 0.
\]
Thus, $T_x(Gx) \oplus \langle R_{1x}, \dots, R_{qx} \rangle \subset (T_x(J^{\Phi-1}(\mu)))^{\perp_\Omega}$. Since $R_{1x}, \dots, R_{qx}$ take values in $T_x(J^{\Phi-1}(\mu))$ but are not tangent to $Gx$, we have
\[
(T_x(J^{\Phi-1}(\mu)))^{\perp_\Omega} = T_x(Gx) \oplus \langle R_{1x}, \dots, R_{qx} \rangle.
\]
\end{proof}

The following theorem is a generalisation of the classical Marsden–Weinstein reduction theorem to the $q$-cosymplectic realm. It follows the ideas of the proof given in \cite{Albert}.

   	  \begin{theorem}\label{MWT}
Let $\Phi:G\times M\rightarrow M$ be a $q$-cosymplectic restricted Hamiltonian action  on the $q$-cosymplectic manifold $(M,\Omega,\vec{\lambda})$ associated with a $q$-cosymplectic momentum map $J^\Phi:M\rightarrow \mathfrak{g}^*$. Assume that $\mu\in\mathfrak{g}^*$ is a weakly regular value of $J^\Phi$ and that $J^{\Phi-1}(\mu)$ is quotientable, i.e. $M_\mu^\Delta := J^{\Phi-1}(\mu)/G_\mu^\Delta$ is a manifold and $\pi_\mu:J^{\Phi-1}(\mu)\rightarrow M_\mu^\Delta$ is a submersion. Let $i_\mu:J^{\Phi-1}(\mu)\hookrightarrow M$ be the natural immersion and let $\pi_\mu: J^{\Phi-1}(\mu)\rightarrow M_\mu^\Delta$ be the canonical projection. Then, there exists a unique $q$-cosymplectic manifold $(M_\mu^\Delta,\Omega_\mu,\vec{\lambda}_\mu)$ such that
\begin{align}\label{iOme}
	i_\mu^*\Omega = \pi^*_\mu\Omega_\mu, \quad i_\mu^*\lambda_i = \pi^*_\mu\lambda_{i\mu}.
\end{align}
\end{theorem}

\begin{proof}
Since $J^{\Phi-1}(\mu)$ is quotientable, the quotient space $M^\Delta_\mu = J^{\Phi-1}(\mu)/G^\Delta_\mu$ is a manifold and $\pi_\mu$ is a surjective submersion. Then, $\ker T\pi_\mu$ is a subbundle of $T(J^{\Phi-1}(\mu))$. 

From the hypothesis that $\Phi_g$ is a $q$-cosymplectomorphism for all $g\in G$, it follows that $L_{\xi_M}\Omega = 0$ and $L_{\xi_M}\lambda_i = 0$ for all $\xi \in \mathfrak{g}$ and $i=1,\dots,q$. Therefore, 
\[
L_{\xi_{J^{\Phi-1}(\mu)}} i^*_\mu \Omega = 0, \quad L_{\xi_{J^{\Phi-1}(\mu)}} i^*_\mu \lambda_i = 0,
\]
for all $\xi \in \mathfrak{g}^\Delta_\mu$ and $i=1,\dots,q$, where $\xi_{J^{\Phi-1}(\mu)}$ denotes the fundamental vector field of the restricted $G^\Delta_\mu$-action on $J^{\Phi-1}(\mu)$.

Now let $Y_{J^{\Phi-1}(\mu)}$ be any vector field tangent to $J^{\Phi-1}(\mu)$. There exists a vector field $Y$ on $M$ such that $Y|_{J^{\Phi-1}(\mu)} = Y_{J^{\Phi-1}(\mu)}$. Then:
\begin{align*}
	i_{Y_{J^{\Phi-1}(\mu)}} i_{\xi_{J^{\Phi-1}(\mu)}} i^*_\mu \Omega &= i^*_\mu(i_Y i_{\xi_M} \Omega) = i^*_\mu(i_Y dJ_\xi) = 0, \\
	i_{Y_{J^{\Phi-1}(\mu)}} i_{\xi_{J^{\Phi-1}(\mu)}} i^*_\mu \lambda_i &= i^*_\mu(i_Y i_{\xi_M} \lambda_i) = 0, \quad i=1,\dots,q.
\end{align*}

These conditions guarantee the existence of forms $\Omega_\mu \in \Omega^2(M^\Delta_\mu)$ and $\lambda_{i\mu} \in \Omega^1(M^\Delta_\mu)$ satisfying \eqref{iOme}. Uniqueness follows from the injectivity of $\pi^*_\mu$ on basic forms.

Since $\Omega$ and $\lambda_i$ are closed, and \eqref{iOme} holds, it follows that $\Omega_\mu$ and $\lambda_{i\mu}$ are also closed.

Next, since $i_{R_i} dJ_\xi = 0$ for all $i=1,\dots,q$ and all $\xi \in \mathfrak{g}$, it follows that $R_i$ is tangent to $J^{\Phi-1}(\mu)$. Let $\tilde R_i := R_i|_{J^{\Phi-1}(\mu)}$; then $\Phi_{g*} \tilde R_i = \tilde R_i$ and $L_{\xi_{J^{\Phi-1}(\mu)}} \tilde R_i = 0$ for all $g \in G^\Delta_\mu$ and $\xi \in \mathfrak{g}^\Delta_\mu$.

Therefore, the pushforwards $R_{i\mu} := \pi_{\mu*} \tilde R_i$ are well-defined vector fields on $M^\Delta_\mu$, and we have:
\begin{align*}
	\pi^*_\mu(i_{R_{i\mu}} \lambda_{i\mu}) &= i_{\tilde R_i} \pi^*_\mu \lambda_{i\mu} = i^*_\mu(i_{R_i} \lambda_i) = 1, \quad i=1,\dots,q, \\
	[R_{i\mu}, R_{j\mu}] &= \pi_{\mu*} [R_i, R_j] = 0, \\
	\pi^*_\mu(i_{R_{i\mu}} \Omega_\mu) &= i_{\tilde R_i} \pi^*_\mu \Omega_\mu = i^*_\mu(i_{R_i} \Omega) = 0.
\end{align*}
So $i_{R_{i\mu}} \lambda_{i\mu} = 1$ and $i_{R_{i\mu}} \Omega_\mu = 0$.

To complete the proof, we show that 
\[
b_\mu : X_\mu \in TM^\Delta_\mu \mapsto i_{X_\mu} \Omega_\mu + \sum_{i=1}^q (i_{X_\mu} \lambda_{i\mu}) \lambda_{i\mu} \in T^* M^\Delta_\mu
\]
is an isomorphism. 

To prove injectivity, suppose $X_\mu$ lies in $\ker b_\mu$. Then $i_{X_\mu} \lambda_{i\mu} = 0$ for all $i$ and $i_{X_\mu} \Omega_\mu = 0$. Let $\tilde X$ be a vector field on $J^{\Phi-1}(\mu)$ projecting to $X_\mu$, i.e., $\pi_{\mu*} \tilde X = X_\mu$, and extend it to $X$ on $M$ with $X|_{J^{\Phi-1}(\mu)} = \tilde X$. Then,
\[
\pi^*_\mu(i_{X_\mu} \Omega_\mu) = i_X \Omega|_{J^{\Phi-1}(\mu)} = 0,
\]
which implies that $X$ takes values in $(T_x J^{\Phi-1}(\mu))^{\perp_\Omega} = T_x(Gx) \oplus \langle R_{1x}, \dots, R_{qx} \rangle$ for each $x \in J^{\Phi-1}(\mu)$. Thus, $X_x = (\xi_M)_x + \sum_{i=1}^q \eta_i R_{ix}$ for some $\xi \in \mathfrak{g}^\Delta_\mu$ and $\eta_i \in \mathbb{R}$ depending on $x$.

Since $\lambda_{i\mu}(T_x\pi_\mu X_\mu) = 0$, we get $\eta_i = 0$ for all $i$, and then
\[
(X_\mu)_{\pi_\mu(x)} = T_x\pi_\mu X_x = T_x\pi_\mu(\xi_M)_x = 0.
\]
So $\ker b_\mu = \{0\}$ and $b_\mu$ is injective.

Finally, since $\dim TM^\Delta_\mu = \dim T^* M^\Delta_\mu$, $b_\mu$ is surjective as well. Therefore, $(M_\mu^\Delta, \Omega_\mu, \vec{\lambda}_\mu)$ is a $q$-cosymplectic manifold.
\end{proof}
In addition to the main results discussed above, it is also of interest to examine 
how restricted Hamiltonian actions with momentum maps interact with local gradient 
vector fields in a more general setting. The following theorem provides a separate 
criterion that, while independent of the main theorem, sheds light on the symmetry 
properties of such actions.

\begin{theorem}
	A restricted Hamiltonian action, with momentum map $J,$ is a symmetry of a local gradient vector field $X$ iff for each $\xi\in\mathfrak{g}$ it holds that $X(J_\xi)=l_\xi$ for some constant $l_\xi\in\mathbb R,$ i.e., 
	$$[X,\xi_M]=0\quad \mathrm{iff}\quad X(J_\xi)=l_\xi.$$
	Specially, when $G$ is a semisimple Lie group, then 
		$$[X,\xi_M]=0\quad \mathrm{iff}\quad X(J_\xi)=0.$$
\end{theorem}
\begin{proof}
	Let vector field $X$ be a local gradient vector field, i.e., it satisfies 
	$$ d\left(i_X\Omega+\sum_{i=1}^q\lambda_i(X)\lambda_i\right)=0.$$
	Then we know that 
	\begin{align*}
		dX(J_\xi)&=di_Xi_{\xi_M}\Omega\\
        &=L_X(i_{\xi_M}\Omega)-i_Xdi_{\xi_M}\Omega\\
        &=L_X(i_{\xi_M}\Omega)\\
		&=i_{[X,\xi_M]}\Omega+i_{\xi_M}L_X\Omega\\
		&=i_{[X,\xi_M]}\Omega-i_{\xi_M}d\left(\sum_{i=1}^q(i_X\lambda_i)\lambda_i\right)\\
		&=i_{[X,\xi_M]}\Omega-i_{\xi_M}d\left(\sum_{i=1}^q(i_X\lambda_i)\lambda_i\right)-di_{\xi_M}\left(\sum_{i=1}^q(i_X\lambda_i)\lambda_i\right)\\
		&=i_{[X,\xi_M]}\Omega-L_{X_M}\left(\sum_{i=1}^q(i_X\lambda_i)\lambda_i\right)\\
			&=i_{[X,\xi_M]}\Omega+\left(\sum_{i=1}^q(i_{[X,X_M]}\lambda_i)\lambda_i\right),
	\end{align*}
	which means that $[X,\xi_M]=0$ iff $dX(J_\xi)=0.$
	
	Since the map $\mathfrak{g}\rightarrow\mathfrak{X}(M),\xi\rightarrow\xi_M$ is a Lie algebra anti-homomorphism, i.e., $[\xi,\eta]_M=-[\xi_M,\eta_M]$,
	given any $\xi,\eta\in\mathfrak{g},$ one can easily calculate that 
	\begin{align*}
		l_{[\xi,\eta]}&=X(J_{[\xi,\eta]})=L_X(J_{[\xi,\eta]})\\
		&=i_XdJ_{[\xi,\eta]}=i_Xi_{[\xi,\eta]_M}\Omega\\
		&=-i_Xi_{[\xi_M,\eta_M]}\Omega\\
		&=-i_XL_{\xi_M}i_{\eta_M}\Omega\\
		&=-i_Xd\xi_M(J_\eta)\\
		&=-L_X(i_{\xi_M}dJ_\eta)\\
		&=-i_{[X,\xi_M]}dJ_\eta-i_{\xi_M}d(X(J_\eta))\\
		&=-i_{[X,\xi_M]}dJ_\eta-i_{\xi_M}dl_\eta=0,
	\end{align*}
	as $[X,\xi_M]=0$ and $l_\eta$ is a constant.
	From this, we infer that in the case where $\mathfrak{g}$ is such that $[\mathfrak{g},\mathfrak{g}]=\mathfrak{g}$, in particular, when $G$ is a semisimple Lie group, the constants $l_\xi,\xi\in\mathfrak{g}$ defined in this theorem are all zero.
\end{proof}

\section{Application: Fast-slow dynamical system}

A \emph{fast--slow system} is a dynamical system with state variables evolving on widely separated time scales. 
In its simplest form, it is written as
\[
\begin{cases}
\dot{x} = f(x,y,\varepsilon), & x \in \mathbb{R}^m \quad \text{(fast variables)}, \\[6pt]
\dot{y} = \varepsilon\, g(x,y,\varepsilon), & y \in \mathbb{R}^n \quad \text{(slow variables)},
\end{cases}
\qquad 0<\varepsilon\ll 1.
\]
Here the \emph{fast variables} $x$ evolve on the natural time $t$, while the \emph{slow variables} $y$ drift only at rate $\mathcal{O}(\varepsilon)$ and reveal their dynamics over long times $t=\mathcal{O}(\varepsilon^{-1})$. 
At $\varepsilon=0$, the slow variables freeze and the fast subsystem decouples. 
Over finite but small $\varepsilon$, averaging methods and adiabatic invariants provide an effective description of the slow drift by replacing the fast oscillations with their mean influence.
Some examples of fast-slow systems appear in celestial mechanics: planets orbit quickly (fast angles), but what matters for centuries are the slow drifts of orbital elements (eccentricity, inclination), or in plasma physics: charged particles gyrate rapidly around magnetic field lines, but the slow drift of their guiding centers determines transport and confinement.

Here we will depict a classical example of fast-slow system, which is a harmonic oscillator with a varying frequency and an extra perturbative term.

\subsection{Fast-slow harmonic oscillator}
Consider the manifold $M = \mathbb{R}_t \times \mathbb{R}_\tau \times T^*Q_f \times T^*Q_s$ with coordinates \((t,\tau; q,p; Q,P)\). We define the forms:
\[
\lambda_1 = dt, \quad \lambda_2 = d\tau, \quad 
\Omega = dq \wedge dp + dQ \wedge dP.
\]
Then
\[
d\lambda_i = 0, \quad d\Omega = 0,
\]
the Reeb vector fields are \(R_1 = \partial_t\), \(R_2 = \partial_\tau\),
\[
\ker \Omega = \mathrm{span}\{R_1, R_2\}, \quad
\xi = \ker \lambda_1 \cap \ker \lambda_2,
\]
and \(\Omega|_{\xi}\) is symplectic. 
Trivially, \(d\lambda_1=d^2t=0\), \(d\lambda_2=d^2\tau=0\), and
\(d\Omega=d(dq\wedge dp + dQ\wedge dP)=0\). Realize that \(\dim T^*Q_f=\dim T^*Q_s=2\), so that \(\dim M = 2n+q=4+2=6\) with \(n=2\) and \(q=2\). Computing
\begin{align*}
  \Omega^{\,n}\wedge \lambda_1\wedge\lambda_2
  &= \Omega^{\,2}\wedge dt\wedge d\tau \\
  &= (dq\wedge dp + dQ\wedge dP)^{\!2}\wedge dt\wedge d\tau\\
  &= 2\; dq\wedge dp\wedge dQ\wedge dP\wedge dt\wedge d\tau \,\neq\,0,
\end{align*}
which is a nowhere-vanishing volume form. Hence \((M,\Omega,\lambda_1,\lambda_2)\) is 2--cosymplectic. Further, see that the Reeb fields \(R_1,R_2\) are defined by
\begin{equation}
  \iotaop_{R_i}\Omega=0, \qquad \lambda_j(R_i)=\delta_{ij},\quad i,j\in\{1,2\}.
  \label{eq:ReebDef}
\end{equation}
Writing a generic vector field as
\(X = a\,\partial_t + b\,\partial_{\tau} + u\,\partial_q + v\,\partial_p + U\,\partial_Q + V\,\partial_P\), we have \(\iotaop_X\Omega = u\,dp - v\,dq + U\,dP - V\,dQ\).
The equations \(\iotaop_{R_i}\Omega=0\) impose \(u=v=U=V=0\). The conditions on the \(\lambda\)'s then give
\[
  R_1=\partial_t,\qquad R_2=\partial_{\tau}.
\]
They commute and preserve the structure: \([R_1,R_2]=0\), \(\mathcal{L}_{R_i}\Omega=0\), \(\mathcal{L}_{R_i}\lambda_j=0\).
We choose the Hamiltonian
\begin{equation}\label{fastslowH}
H(q,p,Q,P,\tau) = \frac12 p^2 + \frac12 \,\omega(Q)^2 \, q^2
+ \varepsilon \, V(q,Q,P,\tau),
\end{equation}
which is non-autonomous via \(\tau\). This Hamiltonian describes a harmonic oscillator with frequency depending on a slow coordinate \(Q\). It is a textbook model in slow--fast Hamiltonian theory and adiabatic invariants \footnote{See, for example, V.I.~Arnold, \emph{Mathematical Methods of Classical Mechanics}
(Springer GTM~60, 2nd ed., 1989), \S50, Example~1, where \(\omega(Q)\) plays the role
of a slowly varying frequency.}. The perturbation term \(\varepsilon\,V(q,Q,P,\tau)\) is the usual small coupling or modulation 
term that depends on the slow variables and possibly explicitly on the slow time \(\tau\) \footnote{Similar forms appear in J.A.~Sanders, F.~Verhulst, J.~Murdock, 
\emph{Averaging Methods in Nonlinear Dynamical Systems} 
(Springer, 2nd ed., 2007), Chapter~5.}.

Now, let us depict the dynamics. For any \(f\in C^\infty(M)\), the 2--cosymplectic Hamiltonian vector field \(X_f\) is uniquely specified by
\begin{equation}
  \iotaop_{X_f}\Omega = df - (R_1 f)\,\lambda_1 - (R_2 f)\,\lambda_2,
  \qquad \iotaop_{X_f}\lambda_i=0\ (i=1,2).
  \label{eq:XfDef}
\end{equation}
Because \(\iotaop_{X_f}\lambda_i=0\), the \(t\) and \(\tau\) components of \(X_f\) vanish. Writing
\(X_f = u\,\partial_q + v\,\partial_p + U\,\partial_Q + V\,\partial_P\) we have
\(\iotaop_{X_f}\Omega = u\,dp - v\,dq + U\,dP - V\,dQ\). Comparing with the right-hand side of \eqref{eq:XfDef} gives the canonical formulas
\begin{equation}
  u = \pdv{f}{p},\quad v = -\pdv{f}{q},\quad U = \pdv{f}{P},\quad V = -\pdv{f}{Q}.
  \label{eq:XfComponents}
\end{equation}
Note the \(t,\tau\)-derivatives of \(f\) only appear along \(\lambda_i\) on the right of \eqref{eq:XfDef}; they do not contribute to components of \(X_f\) because \(\iotaop_{X_f}\lambda_i=0\). So, we can write the Hamiltonian vector field in coordinates as:

\begin{equation}
  X_H 
  = \pdv{H}{p}\,\partial_q - \pdv{H}{q}\,\partial_p
    + \pdv{H}{P}\,\partial_Q - \pdv{H}{Q}\,\partial_P. 
  \label{eq:VecXH}
\end{equation}
The \emph{2--evolution} field is then
\begin{equation}
  E_{a,b}(H) = a\,\partial_t + b\,\partial_{\tau} + X_H.
  \label{eq:EabAgain}
\end{equation}
If we take \((\cdot)\dot{}=\dv*{}{s}\) along the integral curves of \(E_{a,b}(H)\), then
\begin{equation}
  \dot t=a,\qquad \dot\tau=b,
  \qquad \dot q=\pdv{H}{p},\quad \dot p=-\pdv{H}{q},\quad
  \dot Q=\pdv{H}{P},\quad \dot P=-\pdv{H}{Q}.
  \label{eq:EOM}
\end{equation}
With the concrete choice \(a=1\), \(b=\varepsilon\) and Hamiltonian \eqref{fastslowH}, these become
\begin{align}
  &\dot t=1,\qquad \dot\tau=\varepsilon, \\
  &\dot q = p,\qquad 
   \dot p = -\omega(Q)^2 q - \varepsilon\,\pdv{V}{q}, \\
  &\dot Q = \varepsilon\,\pdv{V}{P},\qquad 
   \dot P = -\tfrac12\,\pdv{}{Q}\big(\omega(Q)^2\,q^2\big) - \varepsilon\,\pdv{V}{Q},
   \\
&E_{\varepsilon}:=E_{1,\varepsilon}(H)=\partial_t+\varepsilon\partial_\tau+X_H.\label{eq:EOMExpanded}
\end{align}
These are the multi-time Hamilton equations with fast clock \(t\) and slow clock \(\tau\). Also, note that in the \(q\)-cosymplectic setting, the evolution vector field \(E_\varepsilon\) fulfills
\[
\iota_{E_\varepsilon} \lambda_1 = 1, \quad
\iota_{E_\varepsilon} \lambda_2 = \varepsilon, \quad
\iota_{E_\varepsilon} \Omega =
dH - R_1(H) \lambda_1 - R_2(H) \lambda_2.
\]
as it is expected from the definition of the $2$-cosymplectic structure.

\subsubsection*{Hamilton equations (in fast time \(t\)).} Using \(\dot{}\) to denote \(d/dt\) (so \(\dot{\tau} = \varepsilon\)), we obtain:
\[
\dot{q} = \frac{\partial H}{\partial p} = p, \quad
\dot{p} = -\frac{\partial H}{\partial q} = -\omega(Q)^2 q - \varepsilon\,\partial_q V,
\]
\[
\dot{Q} = \frac{\partial H}{\partial P} = \varepsilon\,\partial_P V, \quad
\dot{P} = -\frac{\partial H}{\partial Q} = -\omega(Q)\omega'(Q)\,q^2 - \varepsilon\,\partial_Q V,
\]
and
\[
\dot{t} = 1, \quad \dot{\tau} = \varepsilon, \quad
\partial_t H = 0, \quad \partial_\tau H = \varepsilon\,\partial_\tau V.
\]

\noindent
See that the pair \((q,p)\) is a fast oscillator with frequency \(\omega(Q)\), slowly modulated by the slow variables \((Q,P)\).
The slow variables drift on the \(t\)-scale with \(\mathcal{O}(\varepsilon)\) velocity from \(V\), and have an \(\mathcal{O}(1)\) back-reaction term \(-\omega\omega'\,q^2\) from the fast subsystem.

\subsubsection*{Singular limit and averaging}
As \(\varepsilon \to 0\), \(Q\) and \(P\) freeze and \((q,p)\) evolves according to
\[
\ddot{q} + \omega(Q)^2 q = 0.
\]
Define the fast action of the harmonic oscillator:
\[
I = \frac{p^2 + \omega(Q)^2 q^2}{2\,\omega(Q)}.
\]

For fixed \( Q \), the pair \((q,p)\) can be written in action–angle coordinates as
\[
q = \sqrt{\tfrac{2I}{\omega(Q)}} \, \sin \theta, 
\qquad 
p = \sqrt{2I \, \omega(Q)} \, \cos \theta.
\]
Hence, over one fast cycle (i.e., as \(\theta\) goes from \(0\) to \(2\pi\)), the average is
\[
\langle q^2 \rangle = \frac{1}{2\pi} \int_{0}^{2\pi} \frac{2I}{\omega(Q)} \sin^2 \theta \, d\theta
= \frac{I}{\omega(Q)}.
\]

%Over fast cycles, \(\langle q^2\rangle = I / \omega(Q)\).
Averaging over the fast cycle and using the result 
\(\langle q^{2} \rangle = I / \omega(Q)\), we obtain
\[
\langle \dot{Q} \rangle = \varepsilon \,\langle \partial_P V \rangle, \quad
\langle \dot{P} \rangle = -\omega'(Q)\,I - \varepsilon\,\langle \partial_Q V \rangle.
\]
Thus, the leading slow drift is \(-\omega'(Q)I\) (adiabatic back-reaction), with \(\mathcal{O}(\varepsilon)\) corrections from \(V\).

\subsubsection*{Hamilton equations (in slow time \(\tau\)).}
If one prefers to evolve in slow time \(\tau\) instead, divide by \(\dot{\tau} = \varepsilon\):
\[
\frac{dq}{d\tau} = \frac{p}{\varepsilon}, \quad
\frac{dp}{d\tau} = -\frac{\omega(Q)^2}{\varepsilon}q - \partial_q V, \quad
\frac{dQ}{d\tau} = \partial_P V, \quad
\frac{dP}{d\tau} = -\frac{\omega(Q)\omega'(Q)}{\varepsilon}q^2 - \partial_Q V.
\]
This makes the fast/slow separation explicit in \(\tau\).

\subsubsection{Averaging in the $q$--cosymplectic setting and focus on slow drift}

At $\varepsilon=0$ and with $(Q,P,\tau)$ frozen, the horizontal Hamiltonian is the harmonic oscillator with a parameter--dependent frequency $\omega(Q)$,
\begin{equation}
H_0(q,p;Q) = \frac12\,p^2 + \frac12\,\omega(Q)^2\,q^2,
\label{eq:HO-param}
\end{equation}
which is an integrable $1$--dof system. A canonical horizontal change to fast action--angle variables $(I,\theta)$ is defined by
\begin{equation}
I = \frac{p^2+\omega(Q)^2 q^2}{2\omega(Q)},
\qquad
\theta = \arg\!\big(\omega(Q) q + i p\big),
\label{eq:Itheta-qcosympl}
\end{equation}
which has the inverse map 
\begin{equation}
q = \sqrt{\frac{2I}{\omega(Q)}}\,\sin\theta, 
\quad 
p = \sqrt{2I\,\omega(Q)}\,\cos\theta.
\label{eq:qp-from-Itheta}
\end{equation}
A direct computation shows that, modulo terms in $\Span\{dQ\}$, 
\[
dq\wedge dp = dI\wedge d\theta,
\]
and, after the standard horizontal Moser correction \cite{DeLeonMarrero2010}, the full $2$--form becomes
\begin{equation}
\omega = dI\wedge d\theta + dQ\wedge dP
\label{eq:omega-canonical-horiz}
\end{equation}
in adapted Liouville coordinates. The Reeb fields remain $R_1=\partial_t$ and $R_2=\partial_\tau$.

In these coordinates, $H_0=\omega(Q)\,I$, so the horizontal Hamilton equations at $\varepsilon=0$ are
\[
\dot{\theta} = \partial_I H_0 = \omega(Q), 
\quad
\dot{I} = -\partial_\theta H_0 = 0,
\]
showing that $\theta$ evolves quickly at frequency $\omega(Q)$, while $I$ is conserved. For $\varepsilon>0$,
\[
H(I,\theta,Q,P,\tau) = \omega(Q)\,I + \varepsilon\,V\big(q(I,\theta,Q),Q,P,\tau\big),
\]
and the horizontal part of $E_{1,\varepsilon}(H)$ is
\begin{equation}
\dot{\theta} = \omega(Q) + \varepsilon\,\partial_I V,
\qquad 
\dot{I} = -\varepsilon\,\partial_\theta V,
\label{eq:Itheta-eqs}
\end{equation}
while the slow variables $(Q,P)$ evolve with
\[
\dot{Q} = \partial_P H, \quad \dot{P} = -\partial_Q H.
\]
The vertical part of $E_{1,\varepsilon}(H)$ is \[\dot{\tau} = \varepsilon, \quad \dot{t}=1.\]

The averaging method (see \cite{ArnoldKozlovNeishtadt, SandersVerhulstMurdock}) can be  applied in this $q$--cosymplectic setting because:  
(i) the Reeb flows $R_1,R_2$ are untouched by horizontal canonical transformations,  
(ii) the horizontal dynamics $(I,\theta,Q,P)$ is symplectic with form \eqref{eq:omega-canonical-horiz},  
(iii) averaging over the fast angle $\theta$ commutes with the projection $\mathrm{pr}:(\widehat M,\widehat\omega)\to(M,\omega,\eta^1,\eta^2)$ in the symplectization.  

We define the $\theta$--average of a function $F(I,\theta,Q,P,\tau)$ by
\[
\langle F\rangle(I,Q,P,\tau) := \frac{1}{2\pi} \int_0^{2\pi} F(I,\theta,Q,P,\tau)\,d\theta,
\]
and the averaged Hamiltonian
\begin{equation}
H_{\mathrm{av}}(I,Q,P,\tau) := \omega(Q)\,I + \varepsilon\,\langle V\rangle(I,Q,P,\tau).
\label{eq:Hav-qcosympl}
\end{equation}
By the averaging theorem \cite[Thm.~3.2]{ArnoldKozlovNeishtadt}, there exists a near--identity horizontal canonical transformation (fixing $t,\tau$) that removes all $\theta$--dependence of $H$ up to $O(\varepsilon)$, and such that over times $t=O(\varepsilon^{-1})$,
\[
(I,Q,P) = (I_{\mathrm{av}},Q_{\mathrm{av}},P_{\mathrm{av}}) + O(\varepsilon),
\]
where $(I_{\mathrm{av}},Q_{\mathrm{av}},P_{\mathrm{av}})$ evolve under the \emph{averaged $q$--evolution}:
\[
E_{1,\varepsilon}(H_{\mathrm{av}}) = R_1 + \varepsilon R_2 + X_{H_{\mathrm{av}}}.
\]
In particular,
\[
\dot{I}_{\mathrm{av}} = 0 + O(\varepsilon^2),\qquad 
\dot{\theta}_{\mathrm{av}} = \omega(Q_{\mathrm{av}}) + \varepsilon\,\partial_I\langle V\rangle + O(\varepsilon^2),
\]
so $I$ is an \emph{adiabatic invariant} at first order. The slow drift is governed by
\begin{equation}
\dot{Q} = \partial_P H_{\mathrm{av}}=\epsilon \partial_P\langle V\rangle(I,Q,P,\tau),\qquad 
\dot{P} = -\partial_Q H_{\mathrm{av}}=-\omega'(Q)I-\epsilon \partial_Q\langle V\rangle(I,Q,P,\tau).
\label{eq:slow-drift}
\end{equation}
The vertical part of $E_{1,\varepsilon}(H)$ is still \[\dot{\tau} = \varepsilon, \quad \dot{t}=1.\]

From the geometric viewpoint, averaging in the $q$--cosymplectic setting is natural because it leaves the vertical Reeb sector $(t,\tau)$ unchanged, modifies only the horizontal Hamiltonian vector field $X_H$, and respects the product structure inherent in $\xi$. In applications, the slow subsystem \eqref{eq:slow-drift} is the physically relevant reduced model: the fast oscillations in $\theta$ average out in observables, but the slow drift in $(Q,P)$ accumulates over long times $O(\varepsilon^{-1})$, making it dominant in the macroscopic behaviour.

\subsubsection{Symmetries and reduction}

The flows of \(R_1=\partial_t\) and \(R_2=\partial_{\tau}\) are strict 2--cosymplectic automorphisms: they preserve \(\Omega\) and leave \(\lambda_i\) fixed. They do not change the horizontal dynamics (they play the role of `clocks`).

\paragraph{Case A. Constant frequency \(\omega(Q)\equiv \omega_0\neq0\)}
Let \(G=S^1\) act by elliptic rotations in the \((q,p)\) plane:
\begin{equation}
  \Phi_\varphi(q,p,Q,P,t,\tau) = \big(\ p\cos(\omega_0\varphi) - 1/\omega_0q\sin(\omega_0\varphi),q\cos(\omega_0\varphi) + \omega_0p\sin(\omega_0\varphi),\ Q,P,t,\tau\big).
\end{equation}
Then \(\Phi_\varphi^*(dq\wedge dp)=dq\wedge dp\), so \(\Omega\) and the \(\lambda_i\) are invariant, i.e. \(\Phi_\varphi\) acts by 2--cosymplectomorphisms. The infinitesimal generator is
\(\xi_M = p\,\partial_q -\omega^2_0 q\,\partial_p\). Compute
\[
  \iotaop_{\xi_M}\Omega = \iotaop_{\xi_M}(dq\wedge dp) = p\,dp + \omega_0^2q\,dq = d\Big(\tfrac{1}{2}(q^2+\omega_0^2p^2)\Big).
\]
Thus the (restricted) momentum map \(J: M\to\mathfrak{g}^*\cong\R\) can be taken as
\begin{equation}
  J(q,p,\cdot) = \tfrac{1}{2}(q^2+\omega_0^2p^2) \qquad (\text{defined up to an additive constant}).
\end{equation}
Then, one may perform Marsden--Weinstein Theorem \ref{MWT} at any regular value \(J=c>0\). The level set \(J^{-1}(c)\) is a circle in each \((q,p)\)-fiber; quotienting by \(S^1\) collapses the fast angle and yields the reduced space with coordinates \((Q,P; t,\tau)\) and reduced structure
\begin{equation}
  \Omega_{\mathrm{red}} = dQ\wedge dP, \qquad \lambda_{1,\mathrm{red}} = dt, \qquad \lambda_{2,\mathrm{red}} = d\tau.
\end{equation}
It is clear that the unperturbed energy $H_0=\tfrac{1}{2}p^2+\tfrac{1}{2}\omega_0^2 q^2$ is \(G\)-invariant (i.e. $L_{\xi_M}H_0=0$) Thus, according to the Proposition \ref{PXH}, we know that $$E_{H_0}(J)=\nabla H_0(J)=X_{H_0}(J)=0.$$

We can also see that, the perturbed energy \eqref{fastslowH} is \(G\)-invariant (i.e. $L_{\xi_M}H=0$) if and only if it has the form $$H=\tfrac{1}{2}p^2+\tfrac{1}{2}\omega_0^2 q^2 + \varepsilon\,\tilde V(Q,P,\tau).$$ 
In this case, 
 according to the Proposition \ref{PXH} again, we know that $$E_H(J)=\nabla H(J)=X_H(J)=0.$$ 
Moreover, on the reduced space, $H$ can also be reduced as 
\begin{equation}
  H_{\mathrm{red}}(Q,P,\tau; c) = \omega_0I + \varepsilon\,\tilde V\big(Q,P,\tau\big),
  \qquad I = \frac{c}{\omega_0}.
\end{equation}
In order to illustrate the coexistence of fast and slow dynamics, we simulate a coupled Hamiltonian-type system with two pairs of canonical variables: $(q,p)$ representing the fast oscillator, and $(Q,P)$ describing the slow degrees of freedom. The fast subsystem evolves according to a harmonic oscillator with natural frequency $\omega_{0}=1$, while the slow subsystem is driven by a small perturbation of order $\varepsilon=0.05$. The system is integrated numerically using \texttt{ode45} over a long time horizon $T_{\max}=200$, starting from the initial conditions $(q(0),p(0),Q(0),P(0))=(1,0,1,0)$. This setup allows us to resolve both the rapid oscillations of $q(t)$ and the slow drift of $Q(t)$, and to compare their respective phase portraits in the $(q,p)$ and $(Q,P)$ planes.

% Figure 1: Fast oscillations (zoomed)
\begin{figure}[htbp]
    \centering
    \includegraphics[width=0.4\textwidth]{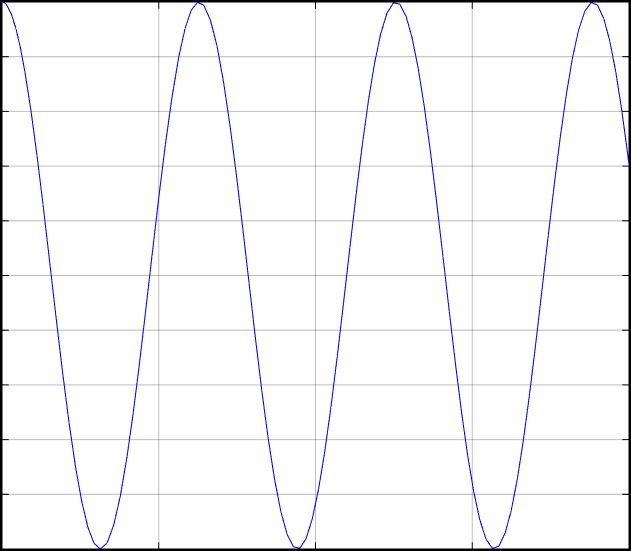}
    \caption{Fast oscillations of the coordinate $q(t)$ shown over a short time window. 
    This illustrates the rapid dynamics at the fast timescale.}
    \label{fig:fs1}
\end{figure}

% Figure 2: Fast vs Slow time comparison
\begin{figure}[htbp]
    \centering
    \begin{minipage}{0.48\textwidth}
        \centering
        \includegraphics[width=\linewidth]{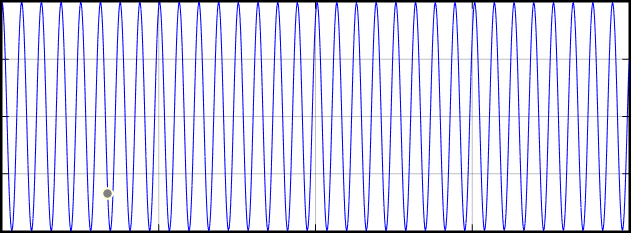}
        \caption*{(a) Fast oscillations $q(t)$ over a long time horizon.}
    \end{minipage}\hfill
    \begin{minipage}{0.48\textwidth}
        \centering
        \includegraphics[width=\linewidth]{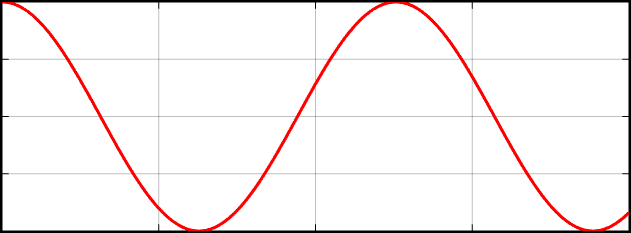}
        \caption*{(b) Slow drift $Q(t)$ over the same time scale.}
    \end{minipage}
    \caption{Comparison of fast and slow time dynamics. 
    (a) The fast variable $q(t)$ oscillates rapidly and appears densely packed over long times. 
    (b) The slow variable $Q(t)$ evolves gradually with velocity of order $\mathcal{O}(\varepsilon)$.}
    \label{fig:fs2combo}
\end{figure}

% Figure 3: Phase portraits (slow vs fast)
\begin{figure}[htbp]
    \centering
    \begin{minipage}{0.48\textwidth}
        \centering
        \includegraphics[width=0.7\linewidth]{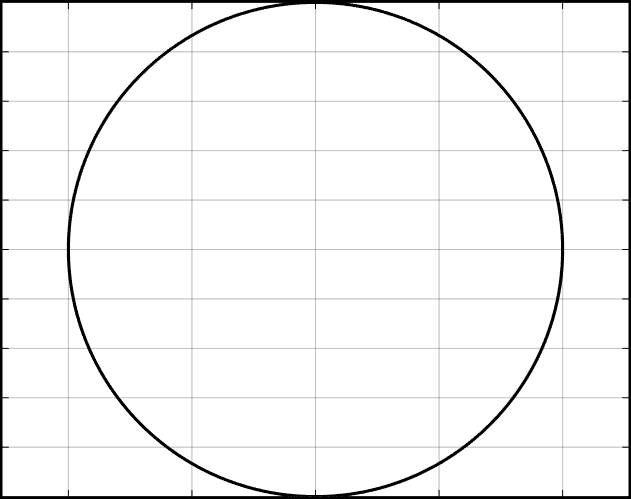}
        \caption*{(a) Slow phase portrait in $(Q,P)$.}
    \end{minipage}\hfill
    \begin{minipage}{0.48\textwidth}
        \centering
        \includegraphics[width=0.7\linewidth]{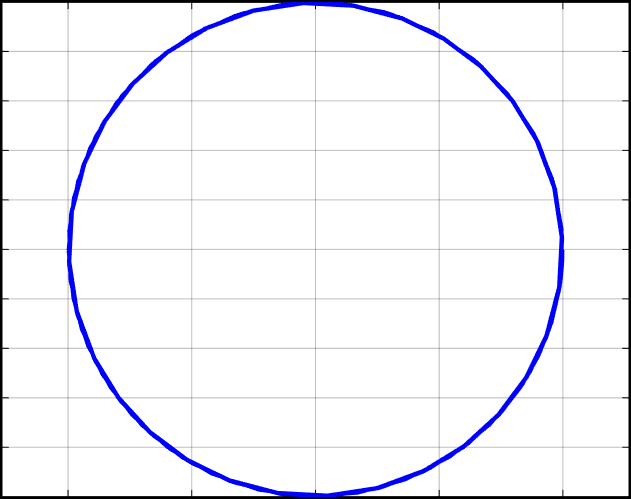}
        \caption*{(b) Fast phase portrait in $(q,p)$.}
    \end{minipage}
    \caption{Geometric representation of the two timescales. 
    (a) The slow subsystem $(Q,P)$ evolves on a circle with frequency proportional to $\varepsilon$. 
    (b) The fast subsystem $(q,p)$ evolves on a circle with frequency $\omega_0$.}
    \label{fig:fs34}
\end{figure}

\begin{figure}[!ht]
    \centering
    
    % Top: q(t)
    \begin{subfigure}{0.9\textwidth}
        \centering
        \includegraphics[width=\linewidth]{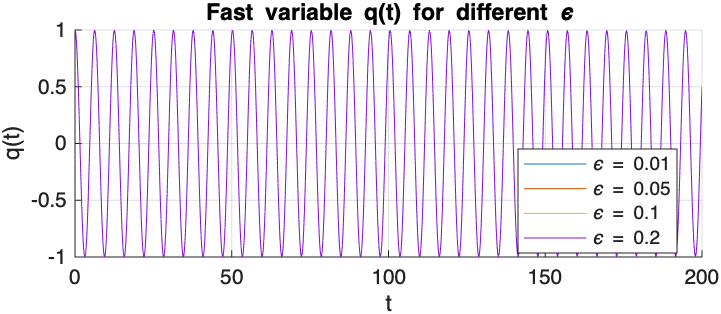}
        \caption{Fast variable $q(t)$ for different $\varepsilon$.}
        \label{fig:q_epsilon}
    \end{subfigure}
    
    \vspace{0.5cm} % spacing between plots
    
    % Bottom: Q(t)
    \begin{subfigure}{0.9\textwidth}
        \centering
        \includegraphics[width=\linewidth]{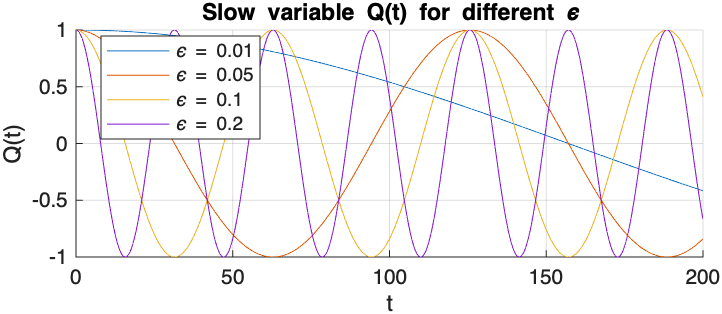}
        \caption{Slow variable $Q(t)$ for different $\varepsilon$.}
        \label{fig:Q_epsilon}
    \end{subfigure}
    
    \caption{Comparison of the fast and slow variables for different values of $\varepsilon$.}
    \label{fig:epsilon_comparison}
\end{figure}

Right above we are plotting the oscillations of both fast and slow variables. In the top panel, the oscillations of $q(t)$ occur essentially at the same frequency $\omega_0$, 
independent of $\varepsilon$. For small time windows, all curves look almost identical. 
The only difference is a very slow modulation of the oscillation amplitude, which becomes 
more visible as $\varepsilon$ increases. Thus, $q(t)$ represents the fast scale of the system 
and is only weakly affected by $\varepsilon$.

In the bottom panel, the dynamics of $Q(t)$ are directly governed by $\varepsilon$. For small values 
(e.g. $\varepsilon = 0.01$), $Q(t)$ changes very little over the entire interval, appearing almost constant. As $\varepsilon$ grows (e.g. $\varepsilon = 0.1, 0.2$), 
the oscillations of $Q(t)$ become faster and more pronounced. In fact, the timescale of these slow oscillations is proportional to $1/\varepsilon$, so doubling $\varepsilon$ 
doubles the speed of the drift in $(Q,P)$.

This makes clear the typical behavior of a fast--slow system: a fast carrier oscillation 
modulated by a slow envelope that depends on $\varepsilon$.

\paragraph{Case B. Variable frequency \(\omega=\omega(Q)\neq0\)}
The elliptic \((q,p)\)-rotation above no longer preserves the form \( \Omega=dq\wedge dp+dQ\wedge dP.\) We consider a rotation of the following form. 

Let \(G=S^1\) act by a rotations $\Phi_t$ in the \((q,p,Q,P,t,\tau)\) space:

\[
\begin{aligned}
q &\;\longrightarrow\;
q\cos\big(\omega(Q)\,t\big)
+ \frac{p}{\omega(Q)}\sin\big(\omega(Q)\,t\big),\\[6pt]
p &\;\longrightarrow\;
p\cos\big(\omega(Q)\,t\big)
- \omega(Q)\,q\sin\big(\omega(Q)\,t\big),\\[6pt]
Q &\;\longrightarrow\;
Q,\\[6pt]
P &\;\longrightarrow\;
P - \omega(Q)\,\omega'(Q)
\int_{0}^{t}
\left(
q\cos\big(\omega(Q)\,s\big)
+ \frac{p}{\omega(Q)}\sin\big(\omega(Q)\,s\big)
\right)^{2}ds,\\[6pt]
t &\;\longrightarrow\; t,\\[6pt]
\tau &\;\longrightarrow\;\tau.
\end{aligned}
\]
Then \(\Phi_t^*(dq\wedge dp+dQ\wedge dP)=dq\wedge dp+dQ\wedge dP\), so \(\Omega\) and the \(\lambda_i\) are invariant, i.e. \(\Phi_t\) acts by 2--cosymplectomorphisms. The infinitesimal generator is
\(\xi_M =p\,\frac{\partial}{\partial q}
- \omega^{2}(Q)\,q\,\frac{\partial}{\partial p}
- \omega(Q)\,\omega'(Q)\,q^{2}\,\frac{\partial}{\partial P}\). Compute
\[
  \iotaop_{\xi_M}\Omega = \iotaop_{\xi_M}(dq\wedge dp+dQ\wedge dP) = p\,dp + \omega^2(Q)q\,dq +\omega(Q)\omega'(Q)q^2dQ= d\Big(\tfrac{1}{2}(q^2+\omega^2(Q)p^2)\Big).
\]
Thus the (restricted) momentum map \(J: M\to\mathfrak{g}^*\cong\R\) can be taken as
\begin{equation}
  J(q,p,\cdot) = \tfrac{1}{2}(q^2+\omega^2(Q)p^2) \qquad (\text{defined up to an additive constant}).
\end{equation}
Then, one may perform Marsden--Weinstein Theorem \ref{MWT} at any regular value \(J=c>0\).  The level set \(J^{-1}(c)\) defines a surface within each \((q,p,Q)\)-fiber. 
Since the coordinates \(t\) and \(\tau\) do not participate in this reduction, 
the reduced space still carries
\[
\lambda_{1,\mathrm{red}} = dt, 
\qquad 
\lambda_{2,\mathrm{red}} = d\tau.
\]

We can also clarify that the unperturbed energy $H_0=\tfrac{1}{2}p^2+\tfrac{1}{2}\omega^2(Q) q^2$ is \(G\)-invariant (i.e. $L_{\xi_M}H_0=0$) Thus, according to the Proposition \ref{PXH}, we know that $$E_{H_0}(J)=\nabla H_0(J)=X_{H_0}(J)=0.$$

Furthermore, it's clear that the perturbed energy \eqref{fastslowH} is \(G\)-invariant (i.e. $L_{\xi_M}H=0$) if and only if $$\omega(Q)\omega'(Q)q^2\frac{\partial V}{\partial P}=p\frac{\partial V}{\partial q}.$$ 
\begin{remark}
   When \(\omega(Q)\) is constant, the transformation and the conclusions discussed in Case~B 
automatically reduce to those in Case~A.
\end{remark}

Let us now plot the dynamics for Case B for some choice of oscillating frequencies and perturbation.

\begin{figure}[htbp]
    \centering
    \begin{subfigure}{0.48\textwidth}
        \centering
        \includegraphics[width=\linewidth]{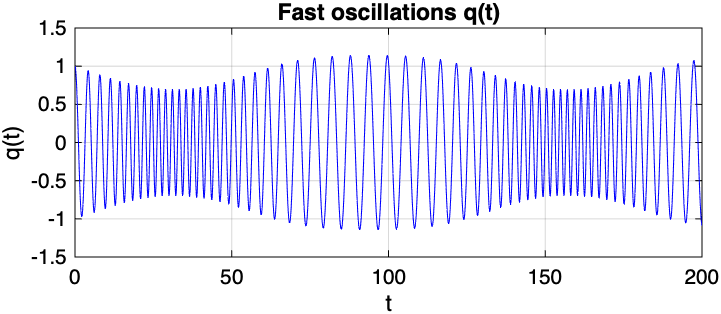}
        \caption{Fast oscillations $q(t)$.}
        \label{fig:caseb1}
    \end{subfigure}\hfill
    \begin{subfigure}{0.48\textwidth}
        \centering
        \includegraphics[width=\linewidth]{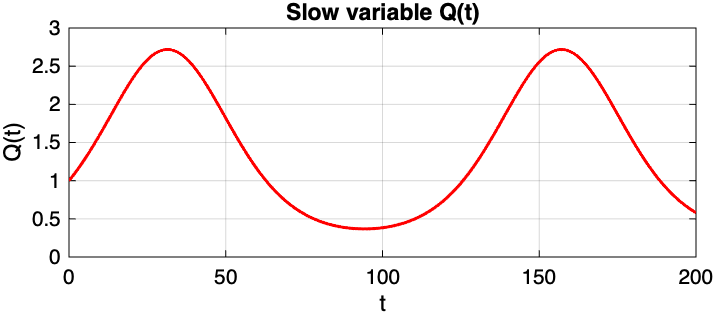}
        \caption{Slow variable $Q(t)$.}
        \label{fig:caseb2}
    \end{subfigure}
    \caption{Case B: fast vs slow dynamics in time domain.}
    \label{fig:caseb12}
\end{figure}

\begin{figure}[htbp]
    \centering
    \begin{subfigure}{0.48\textwidth}
        \centering
        \includegraphics[height=0.35\textheight,keepaspectratio]{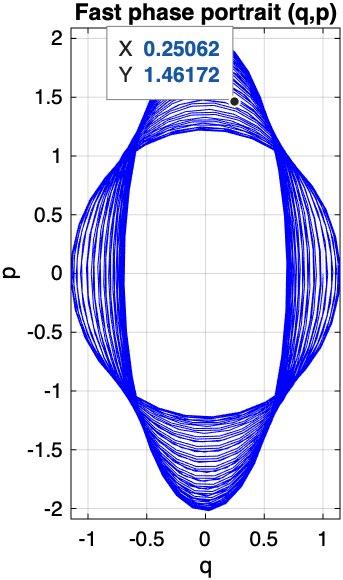}
        \caption{Fast phase portrait $(q,p)$.}
        \label{fig:caseb3}
    \end{subfigure}\hfill
    \begin{subfigure}{0.48\textwidth}
        \centering
        \includegraphics[height=0.35\textheight,keepaspectratio]{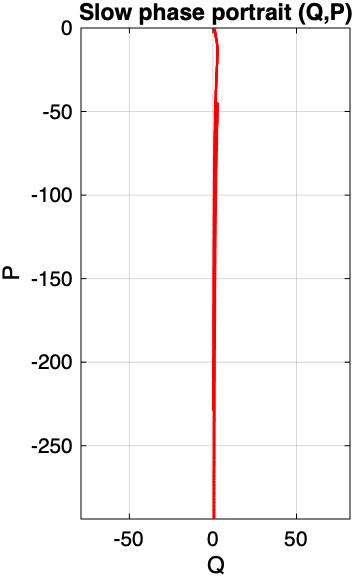}
        \caption{Slow phase portrait $(Q,P)$.}
        \label{fig:caseb4}
    \end{subfigure}
    \caption{Case B: phase portraits for fast and slow subsystems.}
    \label{fig:caseb34}
\end{figure}

\begin{figure}[htbp]
    \centering
    
    \begin{subfigure}{0.9\textwidth}
        \centering
        \includegraphics[width=\linewidth]{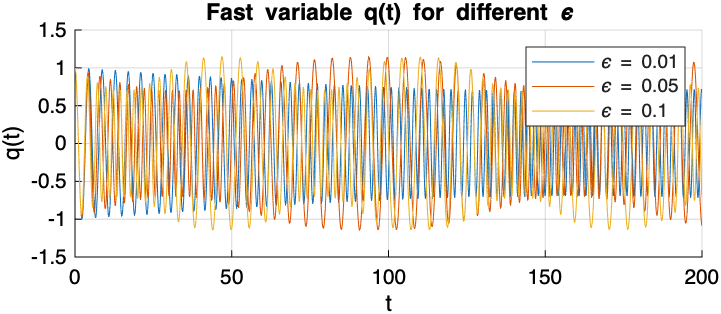}
        \caption{Fast variable $q(t)$ for different $\varepsilon$.}
        \label{fig:caseb5}
    \end{subfigure}
    
    \vspace{0.5cm} % optional vertical spacing
    
    \begin{subfigure}{0.9\textwidth}
        \centering
        \includegraphics[width=\linewidth]{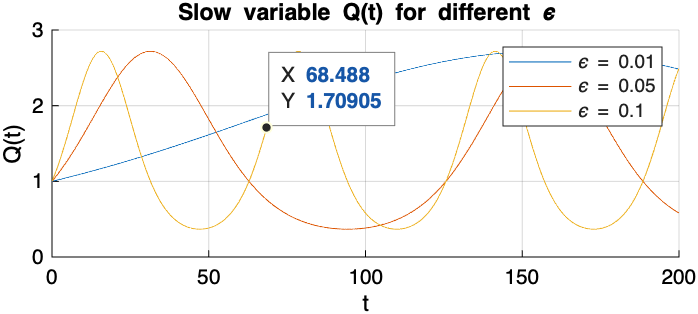}
        \caption{Slow variable $Q(t)$ for different $\varepsilon$.}
        \label{fig:caseb6}
    \end{subfigure}
    
    \caption{Case B: effect of varying $\varepsilon$ on fast and slow variables.}
    \label{fig:caseb56}
\end{figure}

The fast subsystem $(q,p)$ remains nearly elliptic, 
since $q(t)$ oscillates with frequency $\omega(Q)$. 
However, the slow subsystem $(Q,P)$ does not produce closed orbits: 
because of the nonlinear back-reaction term 
$-\omega(Q)\omega'(Q)q^2$, the trajectory in $(Q,P)$ 
is typically a distorted drift rather than an ellipse. 
This reflects the fundamental difference with Case~A 
where $\omega(Q)$ was constant. 

To further highlight the role of $\varepsilon$, 
we compare both the fast variable $q(t)$ 
and the slow variable $Q(t)$ across different values 
of $\varepsilon$ in the last figure. 
The oscillations of $q(t)$ remain fast and nearly unchanged, 
while the drift of $Q(t)$ accelerates proportionally to $\varepsilon$. 
Thus the slow evolution is the most sensitive indicator 
of the perturbation strength.

\paragraph{Averaged dynamics for Case B.} 
In order to illustrate the averaged system, we take the perturbation 
$V(q,Q,P,\tau) = QP\cos(\tau)$ with slow time $\tau = \varepsilon t$. 
After averaging over one fast cycle, the oscillatory terms disappear since 
$\langle \cos(\tau)\rangle = 0$. The resulting reduced equations are therefore 
\[
\dot Q = 0, \qquad \dot P = -\,\omega'(Q)\,I,
\]
where $I=\tfrac{1}{2}(p^2+\omega(Q)^2 q^2)/\omega(Q)$ is the fast action, 
which is conserved at order $\mathcal{O}(\varepsilon)$. 
Thus, the averaged dynamics are extremely simple: the coordinate $Q$ is frozen at 
its initial value $Q(0)=Q_0$, while the momentum $P$ undergoes a uniform drift with slope 
$-\omega'(Q_0)\,I$. In the $(Q,P)$ plane the trajectories collapse to vertical lines.

\vspace{0.3cm}

For numerical illustration we fix the initial condition 
$(q(0),p(0),Q(0),P(0))=(1,0,1,0)$, so that $Q_0=1$ and the fast action is 
$I \approx 0.5$. With the choice $\omega(Q) = \sqrt{1+Q^2}$, its derivative is 
$\omega'(Q)=Q/\sqrt{1+Q^2}$, and therefore the averaged slope is 
$\dot P \approx -\tfrac{1}{\sqrt{2}}\,I \approx -0.35$. 
Consequently, the averaged dynamics produce a straight descent of $P(t)$ 
with $Q(t)$ frozen. 

\begin{figure}[h!]
    \centering
    \begin{subfigure}{0.32\textwidth}
        \includegraphics[width=\linewidth]{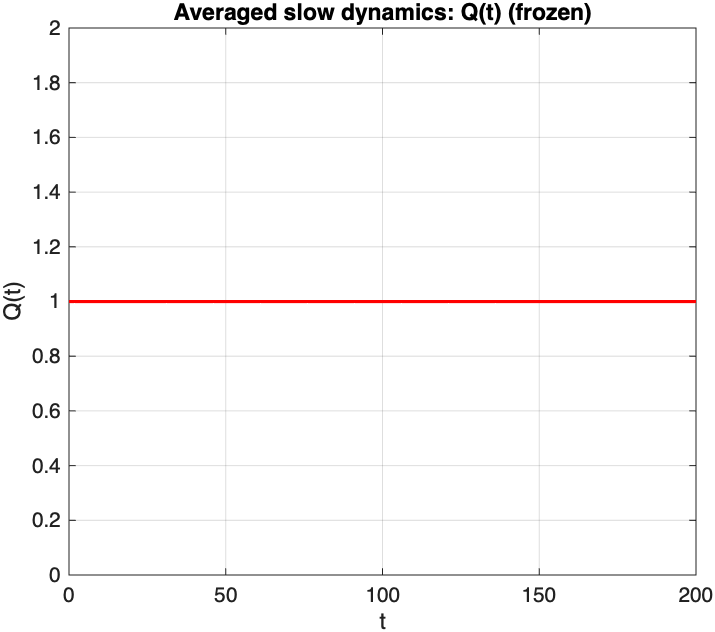}
        \caption{$Q(t)$ frozen}
    \end{subfigure}
    \hfill
    \begin{subfigure}{0.32\textwidth}
        \includegraphics[width=\linewidth]{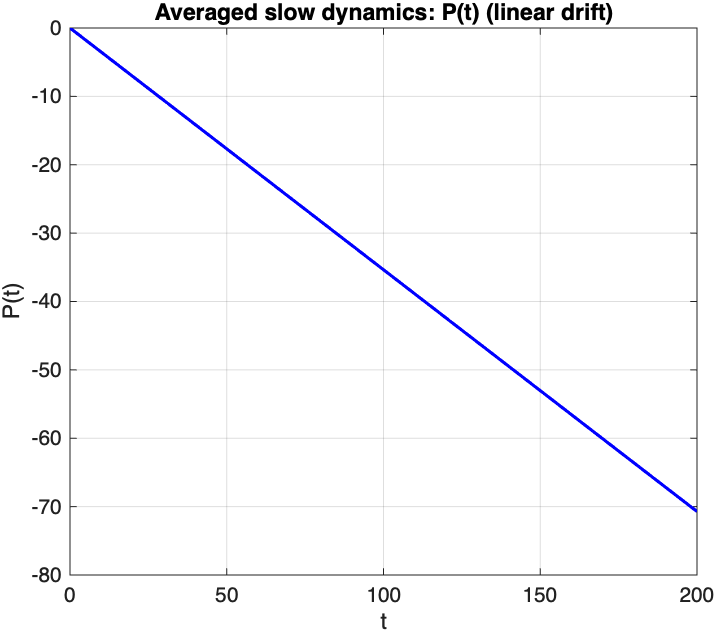}
        \caption{$P(t)$ drift}
    \end{subfigure}
    \hfill
    \begin{subfigure}{0.32\textwidth}
        \includegraphics[width=\linewidth]{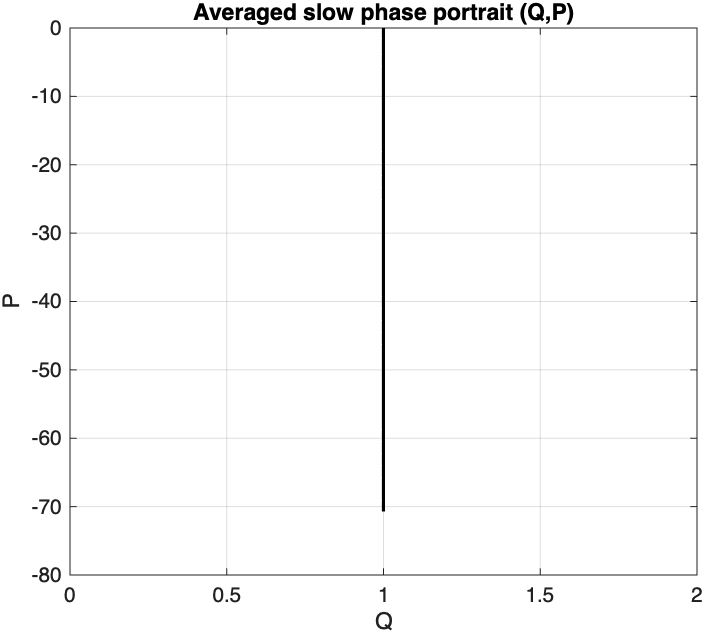}
        \caption{Phase portrait}
    \end{subfigure}
    \caption{Averaged slow dynamics: frozen $Q(t)$, linear drift in $P(t)$, and corresponding $(Q,P)$ phase portrait.}
\end{figure}

Note that in this averaged reduction, the small parameter $\varepsilon$ no longer 
appears explicitly, since the oscillatory forcing averages out. 
The only role of $\varepsilon$ is to control the validity timescale of the 
approximation: the averaged dynamics correctly approximate the full system only 
for times of order $t = \mathcal{O}(\varepsilon^{-1})$. 
In the full (non-averaged) dynamics, the slow variables $(Q,P)$ display oscillatory 
modulations whose frequency is proportional to $\varepsilon$, 
but in the averaged model these oscillations are suppressed, leaving only the 
leading linear drift. 
This contrast highlights the usefulness of averaging: it captures the 
dominant secular behaviour while discarding $\mathcal{O}(\varepsilon)$ fluctuations.

\section*{Acknowledgements}
Cristina Sardón acknowledges Programa Propio of Universidad Politécnica de Madrid for the granting of financial support for research purposes at UCSD the summer of 2025 and gratefully acknowledges Professor Melvin Leok for hosting him at UC San Diego and for his support during this work. Xuefeng Zhao gratefully acknowledges the support from the National Natural Science Foundation of China (Grant No. 12401234).

\end{document}